\documentclass[10pt,pdfa]{article}
\usepackage[in]{fullpage}

\usepackage{url}
\usepackage{hyperref} 
\hypersetup{breaklinks=true}

\usepackage{mdframed}
\usepackage{dsfont}
\usepackage{tabularx}
\usepackage[normalem]{ulem}
\usepackage{comment}
\usepackage[shortlabels]{enumitem}
\usepackage{bm}
\usepackage{bbm}
\usepackage{amsfonts}
\usepackage{xspace}
\usepackage{amsmath}
\usepackage{amssymb}
\usepackage{amsthm}
\usepackage{xcolor}
\usepackage{graphicx}
\usepackage{qtree}
\usepackage{tree-dvips}
\usepackage{float}
\usepackage{physics}
\usepackage{authblk}
\usepackage{braket}
\usepackage[T1]{fontenc}
\usepackage{mathtools}
\usepackage{mathrsfs}
\usepackage{tikz}
\usepackage[linesnumbered,ruled,vlined]{algorithm2e}
\usepackage{qcircuit}
\usepackage[nameinlink,capitalize]{cleveref}

\newcommand{\subversion}[1]{}

  \hypersetup{colorlinks={true},linkcolor={blue},citecolor=magenta}

\newtheorem{theorem}{Theorem}[section]
\newtheorem*{theorem*}{Theorem}
\newtheorem*{corollary*}{Corollary}

\newtheorem{conjecture}[theorem]{Conjecture}

\newtheorem{lemma}[theorem]{Lemma}
\newtheorem{claim}[theorem]{Claim}
\newtheorem{fact}[theorem]{Fact}
\newtheorem{corollary}[theorem]{Corollary}
\theoremstyle{remark}

\theoremstyle{definition}
\newtheorem{definition}[theorem]{Definition}

\numberwithin{equation}{section}
\newcommand\numberthis{\addtocounter{equation}{1}\tag{\theequation}}

\newtheorem*{rep@theorem}{\rep@title}
\newcommand{\newreptheorem}[2]{
\newenvironment{rep#1}[1]{
 \def\rep@title{#2 \ref{##1}}
 \begin{rep@theorem}}
 {\end{rep@theorem}}}
\makeatother

\newreptheorem{theorem}{Theorem}
\newreptheorem{lemma}{Lemma}
\newreptheorem{corollary}{Corollary}

\newcommand{\bit}{\{0,1\}}
\newcommand{\reg}[1]{\mathsf{#1}}
\newcommand{\ignore}[1]{}

\newcommand{\proj}[1]{\ketbra{#1}{#1}}

\newcommand{\E}{\mathop{\mathbb{E}}}
\newcommand{\N}{\mathbb{N}}

\newcommand{\C}{\mathbb{C}}
\newcommand{\negl}{\mathrm{negl}}

\newcommand\id{\mathbbm{1}}
\newcommand{\linear}{\mathrm{L}}
\newcommand{\unitary}{\mathrm{U}}

\newcommand{\bx}{\bm{x}}
\newcommand{\by}{\bm{y}}
\newcommand{\ba}{\bm{a}}
\newcommand{\ind}{\mathbbm{1}}
\newcommand{\supp}{\mathrm{supp}}

\newcommand{\ot}{\otimes}

\newcommand{\poly}{\mathrm{poly}}
\newcommand{\eps}{\epsilon}

\newcommand{\Haar}{\mathrm{Haar}}
\newcommand{\HaarMeasure}{\Haar}
\DeclarePairedDelimiterX{\hsip}[2]{\langle}{\rangle_{\tiny{\mathtt{HS}}}\xspace}{#1, #2}

\newcommand\algo{\mathcal}

\newcommand{\mom}[3]{\cM_{#1}^{(#2)}\left( #3 \right)}
\newcommand{\momhaart}[1]{\mom{{\rm Haar}}{t}{#1}}

\newcommand{\momr}[1]{\mom{R}{t}{#1}}

\newcommand{\mompft}[1]{\mom{PF}{t}{#1}}
\newcommand{\mompfct}[1]{\mom{PFC}{t}{#1}}

\newcommand{\deq}{\coloneqq}

\newcommand{\distinct}{\mathrm{distinct}}

\newcommand{\cA}{\ensuremath{\mathcal{A}}}

\newcommand{\cE}{\ensuremath{\mathcal{E}}}
\newcommand{\cF}{\ensuremath{\mathcal{F}}}

\newcommand{\cH}{\ensuremath{\mathcal{H}}}

\newcommand{\cK}{\ensuremath{\mathcal{K}}}

\newcommand{\cM}{\ensuremath{\mathcal{M}}}
\newcommand{\cN}{\ensuremath{\mathcal{N}}}

\newcommand{\cV}{\ensuremath{\mathcal{V}}}
\newcommand{\cW}{\ensuremath{\mathcal{W}}}
\newcommand{\cX}{\ensuremath{\mathcal{X}}}

\newcommand{\bits}{\{0,1\}}

\title{Simple constructions of linear-depth $t$-designs \\ and pseudorandom unitaries}

\author[1]{Tony Metger}
\author[2]{Alexander Poremba}
\author[3]{Makrand Sinha}
\author[4]{Henry Yuen}
\affil[1]{ETH Zurich}
\affil[2]{Massachusetts Institute of Technology}
\affil[3]{University of Illinois, Urbana-Champaign}
\affil[4]{Columbia University}

\begin{document}
\date{\vspace{-5ex}}

\maketitle

\begin{abstract}
Uniformly random unitaries, i.e.~unitaries drawn from the Haar measure, have many useful properties, but cannot be implemented efficiently.
This has motivated a long line of research into random unitaries that ``look'' sufficiently Haar random while also being efficient to implement.
Two different notions of derandomisation have emerged: 
$t$-designs are random unitaries that information-theoretically reproduce the first $t$ moments of the Haar measure, and pseudorandom unitaries (PRUs) are random unitaries that are computationally indistinguishable from Haar random. 

In this work, we take a unified approach to constructing $t$-designs and PRUs. 
For this, we introduce and analyse the ``$PFC$ ensemble'', the product of a random computational basis permutation $P$, a random binary phase operator $F$, and a random Clifford unitary $C$.
We show that this ensemble reproduces exponentially high moments of the Haar measure.
We can then derandomise the $PFC$ ensemble to show the following: 
\begin{itemize}
\item \textbf{Linear-depth $t$-designs.} We give the first construction of a (diamond-error) approximate $t$-design with circuit depth linear in $t$.
This follows from the $PFC$ ensemble by replacing the random phase and permutation operators with their $2t$-wise independent counterparts.
\item \textbf{Non-adaptive PRUs.}  We give the first construction of PRUs with non-adaptive security, i.e.~we construct unitaries that are indistinguishable from Haar random to polynomial-time distinguishers that query the unitary in parallel on an arbitary state.
This follows from the $PFC$ ensemble by replacing the random phase and permutation operators with their pseudorandom counterparts.
\item \textbf{Adaptive pseudorandom isometries.} We show that if one considers isometries (rather than unitaries) from $n$ to $n + \omega(\log n)$ qubits, a small modification of our PRU construction achieves adaptive security, i.e.~even a distinguisher that can query the isometry adaptively in sequence cannot distinguish it from Haar random isometries. This gives the first construction of adaptive pseudorandom isometries. Under an additional conjecture, this proof also extends to adaptive PRUs.
\end{itemize}
\end{abstract}

\vspace{0.5cm}

\newpage

\section{Introduction} \label{sec:intro}

The Haar measure formalises the notion of a uniformly random unitary and plays an important role in quantum information and computation.
However, Haar random unitaries suffer from a significant practical drawback: they cannot be efficiently implemented.
This has prompted research into ensembles of unitaries that ``look'' sufficiently Haar random for applications while also being efficient to implement.

The same situation arises in classical computer science: here, random functions have many useful properties, but cannot be efficiently implemented.
There exist two different approaches to derandomising a fully random function: $t$-wise independent functions and pseudorandom functions.

$t$-wise independent functions are functions that information-theoretically reproduce the behaviour of uniformly random functions if one only evaluates the function $t$ times.
In other words, even with arbitrary computational power one cannot distinguish $t$ queries to a random $t$-wise independent function from $t$ queries to a uniformly random function.
The drawback of this approach is that one needs to know the number of queries $t$ ahead of time when constructing the functions: if a $t$-wise independent function is evaluated $t+1$ times, its output may look very different from that of a uniformly random function. 

In contrast, pseudorandom functions are indistinguishable from uniformly random functions to any \emph{polynomial-time}\footnote{Whenever we refer to polynomial-time or polynomial numbers of queries, we mean polynomial in the input length $n$ of the function.} distinguisher that can make an arbitrary polynomial number of queries to the function.
The advantage over $t$-wise independent functions is that there is no a priori bound on the number of allowed queries; a pseudorandom function is secure against any polynomial number of queries.
The drawback is that pseudorandom functions are only indistinguishable to computationally bounded distinguishers, whereas $t$-wise independent functions are information-theoretically secure.\footnote{Note that one cannot efficiently construct $t$-wise independent functions where $t$ is superpolynomial in $n$.}

One encounters the same situation when trying to derandomise Haar random unitaries.
The analogue to $t$-wise independent functions are called $t$-designs, which are random unitaries that reproduce the first $t$ moments of the Haar measure~\cite{dankert2009exact,gross2007evenly,ambainis2007quantum}.
The analogue to pseudorandom functions are pseudorandom unitaries (PRUs), which are indistinguishable from Haar random unitaries to any polynomial-time quantum distinguisher with query access to the unitary~\cite{ji2018pseudorandom}.\\

This paper makes three main contributions to the derandomisation of the Haar measure: 
\begin{itemize}
\item \textbf{Linear-depth $t$-designs.} 
It has been a long-standing open question to construct optimal $t$-designs, i.e.~$t$-designs with the smallest possible circuit complexity. 
We give the first construction of a (diamond-error) approximate $t$-design that only requires circuit depth linear in $t$ (and polynomial in the number of qubits, which can be made quasilinear using upcoming work~\cite{random_walks}).
The best prior constructions required depth that scaled quadratically in $t$~\cite{chen2024efficient,haah2024efficient}.

\item \textbf{Non-adaptive PRUs.}  We give the first construction of PRUs with non-adaptive security, i.e.~we construct unitaries that are computationally indistinguishable from Haar random to distinguishers that query the unitary \emph{in parallel} on an arbitrary input state.
Prior work on PRUs imposed severe limitations on the distinguisher, e.g.~only allowing it to query the unitary on i.i.d.~tensor product states~\cite{lu2023quantum}.
\item \textbf{Adaptive Pseudorandom Isometries (PRIs).} We show that if one considers isometries (rather than unitaries) from $n$ to $n + \omega(\log n)$ qubits, a small modification of our PRU construction achieves adaptive security, i.e.~even a distinguisher that can query the isometry adaptively in sequence cannot distinguish it from Haar random isometries. This gives the first construction of adaptive PRIs. Since the number of ancilla qubits introduced by the isometry is only $\omega(\log n)$, this construction may already be a good substitute for adaptive PRUs.
Furthermore, we show that under an additional concrete conjecture our proof also extends to full adaptive PRUs.
\end{itemize}
Our constructions for $t$-designs and PRUs/PRIs are very simple and follow from a more general result: we show that the product of a random Clifford unitary, a random binary phase operator, and a random permutation approximates the Haar measure up to extremely high moments.
The advantage of this random unitary over the Haar measure itself is that it can easily be derandomised:
by replacing the random phase and random permutation with their $t$-wise independent counterparts, we get a linear-depth $t$-design; and by replacing them with their pseudorandom counterparts, we get our PRU/PRI construction.

We now give more detailed background on prior results on $t$-designs and pseudorandom unitaries, before returning to our construction and security proof in \cref{sec:intro-pfc}.

\paragraph{$t$-designs.}
A unitary $t$-design is an ensemble of unitaries that (approximately) reproduces the first $t$ moments of the Haar measure~\cite{dankert2009exact,gross2007evenly,ambainis2007quantum}.
In contrast to Haar random unitaries, $t$-designs on $n$ qubits can be implemented in polynomial time (in $n$ and $t$), making them useful for applications ranging from randomised benchmarking~\cite{magesan2012characterizing,knill2008randomized} to quantum complexity growth~\cite{brandao2021models} and black hole physics~\cite{Hayden_2007}.

Unitary $t$-design constructions are typically approximate, i.e.,~they only reproduce the moments of the Haar measure approximately.
Various notions of approximation have been introduced (see e.g.~\cite{mele2023introduction} for an overview).
For our discussion, we need to consider two different notions of approximation:
diamond-error and relative-error. 
We refer to \cref{def:additive_error_design,def:relative-error-design} for the formal definitions. 
For our discussion here, it suffices to know that diamond-error corresponds to non-adaptive security (i.e.~the $t$-design looks like a Haar random ensemble to a distinguisher making $t$ parallel queries) and relative-error corresponds to (and is in fact stronger than) adaptive security (i.e.~the $t$-design still looks Haar random even if the distinguisher is allowed to make $t$ adaptive queries).
This means that relative-error is a stronger notion, and converting from diamond-error to relative-error incurs a multiplicative $2^{O(nt)}$ blow-up in the error parameter.
For our discussion below, we always require the approximation error (in either notion of approximation) to be negligible in $n$ for any $t = \poly(n)$ and do not write it explicitly.

It is a long-standing open question to construct $t$-designs as efficiently as possible; by ``efficient'' we mean with a minimal number of quantum gates or circuit depth, but randomness-efficient constructions have also been studied~\cite{o2023explicit}.
In a seminal result, Brandao, Harrow and Horodecki~\cite{brandao2016local} showed that random 2-local quantum circuits on $n$ qubits form relative-error $t$-designs in depth $O(n^2 t^{10})$, which was improved to depth $O(n t^{5+o(1)})$ by Haferkamp~\cite{haferkamp2022random}.
In addition to these results on random circuits, very recently two independent works used more structured circuits to achieve more efficient $t$-designs.
Chen, Docter, Xu, Bouland, and Hayden~\cite{chen2024efficient} achieve diamond-error $t$-designs in depth $\tilde O(n^2 t^2)$ using an intricate analysis of exponentiated Gaussian Unitary Ensembles.
Haah, Liu, and Tan~\cite{haah2024efficient} use random Pauli rotations to achieve relative-error $t$-designs in  depth $\tilde O(n t^2)$.

Our construction is very simple and is the first to achieve linear scaling in $t$: we construct diamond-error $t$-designs with circuit depth $O(\poly(n) t)$.
We can also amplify our construction to achieve \emph{relative-error} $t$-designs with circuit depth $O(\poly(n) t^2)$.
The exact polynomial dependence in $n$ depends on the details of a construction of $O(t)$-wise independent permutations due to Kassabov~\cite{Kassabov05}, which is analysed explicitly in an upcoming work~\cite{random_walks}, showing that the depth can also be made quasilinear in $n$.

\paragraph{Pseudorandom unitaries.}

Pseudorandom unitaries (PRUs) are ensembles\footnote{Strictly speaking, a PRU ensemble is an infinite sequence $\mathcal{U} = \{\mathcal{U}_n\}_{n \in \N}$ of $n$-qubit unitary ensembles $\mathcal{U}_n = \{U_k\}_{k \in \mathcal{K}}$, where $n \in \N$ serves as the security parameter (see \Cref{def:PRU} for a formal statement).} $\{U_k\}_{k \in \mathcal{K}_n}$ of unitaries that are efficient to implement, but that look indistinguishable from Haar random unitaries to any polynomial-time distinguisher.
This means that no polynomial-time distinguisher with oracle access to either a Haar random unitary or a unitary chosen uniformly from the PRU ensemble $\{U_k\}$ can tell the two cases apart.
This makes PRUs the natural quantum analogue to pseudorandom functions (PRFs).

The concept of PRUs was introduced by Ji, Liu, and Song~\cite{ji2018pseudorandom}. 
Their paper gave a conjectured construction of PRUs, but only proved security (assuming quantum-secure one-way functions) for a much weaker primitive called \emph{pseudorandom states} (PRSs).
A pseudorandom state ensemble is a set of states (rather than unitaries) that look indistinguishable from Haar random states (even with access to polynomially many copies of the state).
Since their introduction, PRSs have become an influential concept with  applications in quantum cryptography~\cite{ananth2022cryptography,morimae2022quantum}, lower bounds in quantum learning theory~\cite{huang2022quantum}, and even connections to quantum gravity~\cite{aaronson2024quantum}.
However, proving security for a pseudorandom \emph{unitary} construction has remained an open problem, and~\cite{haug2023pseudorandom} even explores restrictions on possible PRU constructions.

Given the difficulty of proving security for PRUs, recent work has considered intermediate steps between PRSs and PRUs.
Lu, Qin, Song, Yao and Zhao~\cite{lu2023quantum} introduced the notion of a pseudorandom state scrambler (PRSS).
A PRSS ensemble is a set of unitaries $\{U_k\}_{k \in \cK}$ such that for all states $\ket{\phi}$, the state family $\{U_k \ket{\phi}\}_{k \in \cK}$ is a PRS.
In other words, a PRSS is a PRU that is only secure when queried on i.i.d.~product states $\ket{\phi}^{\ot \poly(n)}$.
If we restrict the state $\ket{\phi}$ to the all-0 state $\ket{0}$, then we recover PRSs as a special case.
\cite{lu2023quantum} showed how to construct a PRSS ensemble (assuming quantum-secure one-way functions) by an intricate analysis of Kac's random walk.
Ananth, Gulati, Kaleoglu, and Lin~\cite{ananth2023pseudorandom} define the notion of pseudorandom isometries (PRIs), which are like PRUs except that the operation can introduce extra qubits, i.e.~it is an isometry, not a unitary. Towards constructing a PRI ensemble, \cite{ananth2023pseudorandom} show that a candidate construction is computationally indistinguishable from Haar random isometries when applied to certain types of tensor product states (including tensor powers of the same state and i.i.d.~Haar random states).~\cite{ananth2023pseudorandom} also explore the cryptographic applications of PRIs, which we briefly discuss in \cref{sec:discussion}.  

In this work, we give the first construction of PRUs with non-adaptive security and PRIs with adaptive security. For the non-adaptive PRUs, we allow the polynomial-time distinguisher to query the PRU polynomially many times \emph{in parallel} on an \emph{arbitrary entangled state} (rather than the restricted classes of input states considered in prior work).
For the adaptive PRIs, we show that appending a $\omega(\log n)$-qubit auxiliary state to the input (which is what makes this construction an isometry) and applying our PRU construction achieves security even against adaptive adversaries, i.e.~the ensemble of isometries looks computationally indistinguishable from Haar random isometries even to distinguishers that can perform arbitrary sequential adaptive queries to the isometry.

\subsection{A unified approach to $t$-designs and PRUs: the $PFC$ ensemble} \label{sec:intro-pfc}

$t$-designs and PRUs are intimately related and our work takes a unified approach to constructing both.
The key insight of our proofs is that while the Haar measure is difficult to derandomise directly, we can construct a different ensemble of unitaries, which we call the $PFC$ ensemble, that matches exponential-order moments of the Haar measure. In other words, this ensemble is a (diamond-error) $t$-design even when $t$ is exponential in the number of qubits.

The $PFC$ ensemble is not efficient to implement either, but it has a key advantage over the Haar measure: it is easy to derandomise, both in an information-theoretic manner to get $t$-designs and in a computational manner to get PRUs. As such, a significant portion of our paper is concerned with analysing the properties of the $PFC$ ensemble itself, and the results on $t$-designs and (non-adaptive) PRUs follow straightforwardly from classical results on $t$-wise independence and pseudorandom functions. 

The $PFC$ ensemble is the following random ensemble of unitary matrices.

\begin{definition}[$PFC$ ensemble] \label{def:pfc-ensemble}
The $n$-qubit (or $2^n$-dimensional) $PFC$ ensemble is given by the product $PFC$ where $P$ is a uniformly random permutation matrix on $n$-qubit computational basis states, $F: \ket{x} \to (-1)^{f(x)}\ket{x}$ is a diagonal unitary with a uniformly random function $f: \bit^n \to \bit$, and $C$ is a uniformly random $n$-qubit Clifford, all sampled independently. 
\end{definition}

Our main technical result says that the $PFC$ ensemble is a $t$-design for exponential $t$ with negligible diamond distance error.
We refer to \cref{sec:pfc_proof_sketch} for a proof sketch and to \cref{thm:pfc-ensemble} for the formal statement and proof.

\begin{theorem*}[Informal] \label{thm:intro-pfc-ensemble}
For $d = 2^n$, the $d$-dimensional $PFC$ ensemble is a diamond $\epsilon$-approximate $t$-design for $\epsilon = O(t/\sqrt{d})$.
This means that for all $t$ and all states $\ket{\psi}_{\reg{AE}}$, where $\reg{A}$ is an $nt$-qubit register (on which the unitary acts) and $\reg E$ is an arbitrary ancilla register, 
\begin{align*}
\norm{\E_{U\sim \Haar} U^{\ot t}_{\reg A} \proj{\psi} U^{\ot t,\dagger}_{\reg A} - \E_{PFC} (PFC)^{\ot t}_{\reg A} \proj{\psi} (PFC)^{\ot t,\dagger}_{\reg A}}_1 \leq O(t/\sqrt{d})\,.
\end{align*}
\end{theorem*}

This means that even for exponential $t$, e.g.~$t = 2^{n/4}$, the $PFC$ ensemble is a $t$-design with exponentially small (in $n$) diamond distance error. These parameters turn out to be strong enough to give a unified construction of efficient $t$-designs and non-adaptive PRUs, as we will see next. Before doing so, we remark that the only property of a random Clifford $C$ that we require is that it is an exact $2$-design. As such, we can take $C$ to be any exact $2$-design (or even an approximate one with small enough error) in the construction given in \Cref{def:pfc-ensemble}.

\paragraph{Information-theoretic derandomisation: $t$-designs with linear depth.} \label{sec:intro_t_designs}

While \cref{thm:intro-pfc-ensemble} shows that the $PFC$ ensemble is a $t$-design with error $O(t/\sqrt{d})$, this result is not immediately useful as it is not possible to efficiently implement a uniformly random phase operator $F$ or a uniformly random permutation $P$. However, we can replace the (classical) uniform functions and permutations underlying the phase and permutation operators by their $O(t)$-wise independent variants (for a formal definition of $t$-wise independence, see \Cref{sec:linear-t-design}) and show that the resulting random unitary remains a $t$-design (see \Cref{thm:t-designs}).

\begin{theorem*}[Efficient unitary t-design, informal] \label{thm:intro-t-designs}
Let $d=2^n$ and let $\nu$ be the $n$-qubit $PFC$ ensemble where $F$ is instantiated with a random $2t$-wise independent Boolean function, $P$ is instantiated with a random (approximately) $t$-wise independent permutation, and $C$ is a random $n$-qubit Clifford. Then, $\nu$ is a diamond error $\eps$-approximate $t$-design for $\eps = O(t/\sqrt{d})$.
\end{theorem*}

By using efficient constructions of $2t$-wise independent functions and approximate $t$-wise independent permutations~\cite{Kassabov05,v009a015,caprace2023tame,random_walks}, we obtain $t$-designs with circuit size and depth $O(t \; \poly(n))$. 
By repeating the construction multiple times in sequence, we can make the error $\eps$ decay exponentially and also obtain relative-error $t$-designs with circuit size and depth $O(t^2 \cdot \poly(n))$.
More precisely, we get the following (see \cref{cor:t-designs-linear,lem:error_amplified} for the formal statement):
\begin{corollary*}[informal] For any number of qubits $n$, tolerated error $\eps > 0$, and $t \leq 2^{n/4}$, there exists a diamond $\eps$-approximate $t$-design with circuit size and depth $O( t \cdot \poly(n) + t \log1/\eps)$ and a relative-error $\eps$-approximate $t$-design with circuit size and depth $O(t^2 \cdot \poly(n) + t^2 \log1/\eps)$.
\end{corollary*}

In fact, as mentioned before, the upcoming work \cite{random_walks} gives explicit circuits for the $O(t)$-wise independent permutations from~\cite{Kassabov05} that also achieve quasilinear depth in $n$, which implies $t$-designs in depth $\tilde O(tn)$ (for diamond distance error $\eps = e^{-\Omega(n)}$).

\paragraph{Computational derandomisation: non-adaptive pseudorandom unitaries.} \label{sec:intro_pru}

Replacing the uniformly random function and permutation in the $PFC$ ensemble with their pseudorandom counterparts, we immediately get a PRU ensemble with non-adaptive security. More concretely, we require:

\begin{itemize}
    \item An ensemble of (quantum-secure) pseudorandom permutations (PRPs)~\cite{zhandry2016note}. Broadly speaking, this is a family $\{\pi_{k_1} : \{0,1\}^n \to \{0,1\}^n \}_{k_1 \in \mathcal{K}_1}$ of permutations with the property that, for a randomly chosen key $k_1 \sim \mathcal{K}_1$, the permutation $\pi_{k_1}$ is computationally indistinguishable from a perfectly random permutation.
    For a given $\pi_{k_1}$, we let $P_{k_1}$ be the corresponding $n$-qubit permutation matrix.
    \item An ensemble of (quantum-secure) pseudorandom functions (PRFs)~\cite{zhandry2021construct}. More formally, this is a family $\{f_{k_2} : \{0,1\}^n \to \{0,1\} \}_{k_2 \in \mathcal{K}_2}$ with the property that, for a randomly chosen key $k_2 \sim \mathcal{K}_2$, the function $f_{k_2}$ is computationally indistinguishable from a perfectly random Boolean function. For a given $f_{k_2}$, we let $F_{k_2}: \ket{x} \to (-1)^{f_{k_2}(x)} \ket{x}$ be the $n$-qubit phase oracle implementing $f_{k_2}$.
\end{itemize}
Then, indexing the $n$-qubit Clifford group as $\{C_{k_3} \}_{k_3 \in \mathcal{K}_{3}}$, our PRU is the $n$-qubit ensemble $\{U_k\}_{k \in \mathcal{K}}$ which is specified by a key $k = (k_1, k_2, k_3) \in \cK_1 \times \cK_3 \times \cK_3 \eqqcolon \cK$, where
\begin{equation}
    \label{eq:intro-pru}
    U_k = P_{k_1} F_{k_2} C_{k_3}.
\end{equation}
Note that, the Clifford group on $n$-qubits has size $2^{O(n^2)}$.
Thus assuming the PRP and PRF scheme both have a key space consisting of $\mathcal{K}_1 = \mathcal{K}_2 = \bit^n$, our PRU ensemble has a key length\footnote{We note that one can apply a classical pseudorandom random generator to the key to reduce the key length.} of $|k| = n + n + O(n^2) = O(n^2)$, where $n \in \N$ is the security parameter. Because an ensemble of (quantum-secure) PRFs and PRPs can  efficiently be constructed from (quantum-secure) one-way functions~\cite{zhandry2016note,zhandry2021construct}, the security of our PRU ensemble also only relies on the existence of quantum-secure one-way functions. Furthermore, since random Cliffords can be efficiently sampled~\cite{berg2021simple}, there is a polynomial-time quantum algorithm that implements the PRU.

From \cref{thm:intro-pfc-ensemble}, it is not hard to show that the PRU ensemble is secure against any distinguisher that makes polynomially many \emph{parallel} queries to the PRU, i.e.~the distinguisher is allowed a single query to $U^{\ot \poly(n)}$ on \emph{any input state}, rather than just the restricted classes of input states allowed in PRSs and PRSSs.
We give the the formal statement and proof in \cref{sec:pru_formal}.

\begin{theorem*}[Non-adaptive PRUs, informal] 
    Assuming the existence of quantum-secure one-way functions, the ensemble described in~\Cref{eq:intro-pru} satisfies non-adaptive PRU security. 
\end{theorem*}

\paragraph{Pseudorandom isometries with adaptive security.}

We conjecture that our PRU construction is also adaptively secure, but so far we are not able to prove adaptive security for the case of unitaries. However, we can show adaptive security under a small relaxation: if instead of PRUs, we consider PRIs that map $n$ qubits to $n + \omega(\log n)$ qubits, then we can show adaptive security. Our PRI construction is the same as our PRU construction, with two minor differences: we fix part of the input state to the $\ket{+}$-state (which is what makes this an isometry) and we do not require the random Clifford.

\begin{definition}[PRI construction] \label{def:pri_construction}
Let $\{P_{k_1}\}_{k_1 \in \cK_1}$ and $\{F_{k_2}\}_{k_2 \in \cK_2}$ be the pseudorandom permutation and binary phase operators on $n$ qubits from \cref{eq:intro-pru}.
Choose any function $s(n) = \omega(\log n)$.
We then define an ensemble of isometries $\{V_k\}_{k \in \cK_1 \times \cK_2}$ from $n - s(n)$ to $n$ qubits given by 
\begin{align*}
V_{k_1, k_2} \ket{\psi} = P_{k_1} F_{k_2} (\ket{\psi} \otimes \ket{+}^{\ot s(n)}) \,.
\end{align*}
\end{definition}

We show that no polynomial-time quantum algorithm can distinguish such isometries from a Haar random isometry (see \Cref{sec:towards-adaptive} for a definition). We prove this formally in \cref{sec:pri_adaptive_security} and sketch the proof in \cref{sec:intro_adaptive_pri}.

\begin{theorem*}[Adaptive PRIs, informal] \label{thm:pri_adaptive_intro}
 Assuming the existence of quantum-secure one-way functions, the isometries defined in \Cref{def:pri_construction} are pseudorandom isometries with adaptive security. 
\end{theorem*}

Since the number of additional output qubits $\omega(\log n)$ is quite small, it seems likely that this PRI construction already suffices in many situations where one would like to use adaptively secure PRUs. Furthermore, we can also extend the proof of PRI security to full adaptive PRU security assuming a concrete conjecture (see \Cref{sec:towards-adaptive}).

Note that when written as matrices, the isometries $V_k$ from \cref{thm:pri_adaptive_intro} have only real entries. In contrast, it is known that PRUs with only real entries cannot exist~\cite{haug2023pseudorandom}. The fact that (even adaptively-secure) PRIs with real entries are possible may seem surprising at first sight.
However, we note that already in our PRU construction, the only source of complex numbers were the random Cliffords, and their only purpose in our construction is to ensure that the initial input state is sufficiently scrambled so that it has a large overlap with a certain subspace that we call the \emph{distinct subspace} (see \Cref{sec:overview}). Our analysis for \cref{thm:pru_security}  in fact shows that on the distinct subspace, the (real-valued) $PF$-ensemble (with pseudorandom permutations and functions) forms a non-adaptive PRU. This is also closely related to the very recent work of Brakerski and Magrafta~\cite{brakerski2024real}, who construct real-valued unitaries that look Haar random on any polynomial-sized set of orthogonal input states.

\subsection{Proof Overview}
\label{sec:overview}

We break this section into two parts: the first focuses on analyzing the $PFC$ ensemble and the second focuses on adaptive security for PRIs.

\subsubsection{$t$-design property of the $PFC$ ensemble} \label{sec:pfc_proof_sketch}

As mentioned before, showing that the $PFC$ ensemble is a diamond distance $t$-design corresponds to showing that any non-adaptive algorithm (possibly inefficient) that makes $t$ queries cannot distinguish the $PFC$ ensemble from a Haar random ensemble. Such an algorithm starts with an initial state $\ket{\psi}_{\reg{A}_1 \cdots \reg{A}_t \reg{B}}$ where registers $\reg{A}_1,\ldots,\reg{A}_t$ are each on $n$ qubits, and $\reg{B}$ is an arbitrary workspace register. The algorithm then applies the unitary $U^{\otimes t}$ on the registers $\reg{A}_1, \cdots, \reg{A}_t$ (where $U$ is either $PFC$ or a Haar random matrix) and performs a measurement afterwards. In order to show that no such algorithm can distinguish the $PFC$ ensemble from a Haar random ensemble, it suffices to show that for all initial states $\ket{\psi}_{\reg{A}_1 \cdots \reg{A}_t \reg{B}}$, the following two density matrices are close in trace distance:
\begin{equation}
    \label{eqn:security}
        \E_{PFC} ((PFC)^{\ot t} \ot \id) \proj{\psi} ((PFC)^{\ot t} \ot \id)^\dagger \approx \E_{U \sim \Haar} (U^{\ot t} \ot \id) \proj{\psi} (U^{\ot t} \ot \id)^\dagger~.
\end{equation}

Here, $\ket{\psi}$ is a completely arbitrary state: the registers $\reg{A}_1,\ldots,\reg{A}_t$ of $\ket{\psi}$ may all be entangled with each other, the reduced states on each register may be different from each other, and finally the distinguisher is allowed access to the purification of the input state to the unitaries. 

For this proof overview, we assume that our algorithm has no workspace, i.e.~the initial state is $\ket{\psi}_{\reg{A}_1\cdots\reg{A}_t}$. This is merely to simplify the notation, the proof works exactly in the same way even in the case of an entangled workspace. It then suffices to show that the following two density matrices are close in trace distance:
\begin{equation}
    \label{eqn:inf-security}
      \E_{PFC} (PFC)^{\ot t} \proj{\psi} ((PFC)^{\ot t})^\dagger \approx \E_{U \sim \Haar} U^{\ot t}  \proj{\psi} (U^{\ot t})^\dagger~.
\end{equation}

For this, we leverage Schur-Weyl duality to compute what these two density matrices explicitly look like. Let $d=2^n$ be the dimension.
According to Schur-Weyl duality:
\begin{enumerate}[(1)]
    \item the space $(\C^d)^{\ot t}$ can be decomposed as a direct sum\footnote{Here $\lambda \vdash t$ means that $\lambda$ is a partition of $[t]$.} $\bigoplus_{\lambda \vdash t} P_\lambda$ where $P_\lambda = {W_\lambda \ot V_\lambda}$ is a tensor product of two spaces, and
    \item any unitary $U^{\otimes t}$ only acts non-trivially on the subspaces $W_{\lambda}$ and any unitary $R_\pi$ that permutes the $t$ subsystems according to a permutation $\pi \in S_t$ acts non-trivially only on the subspaces $V_{\lambda}$.
\end{enumerate}

Using this, we show that applying a \emph{$t$-wise Haar twirl} to $\ket{\psi}$ to obtain the state on the right hand side of \cref{eqn:inf-security} results in the following state:

\begin{equation*}
    \E_{U \sim \Haar} U^{\ot t}  \proj{\psi} (U^{\ot t})^\dagger = \sum_{\lambda \vdash t} \frac{\ind_{W_\lambda}}{\Tr[\ind_{W_\lambda}]} \otimes \Tr_{W_\lambda}[\ind_{P_\lambda} \ketbra{\psi}{\psi}\ind_{P_\lambda}]\,,
\end{equation*}
where $\ind_{W_{\lambda}}$ is the identity on the subspace $W_\lambda$, i.e.~the orthogonal projection onto that subspace.
In particular, the state is a direct sum of tensor product states where the state on the subspaces $W_{\lambda}$ is maximally mixed. 

Next, we would like to show that applying a \emph{$t$-wise $PFC$ twirl} to $\ket{\psi}$ to obtain the left hand side state in \cref{eqn:inf-security} results in a state that is close to the above. For technical reasons, the computations here are easier if the state $\ket{\psi}$ is supported only on the subspace of distinct computational basis states in the registers $\reg{A}_1, \ldots \reg{A}_t$, i.e., the subspace spanned by $\ket{x_1,\ldots,x_t}$ where $x_1,\ldots, x_t \in \bits^n$ are all distinct. We show that applying a random Clifford ensures this by showing that 
\[\Tr[\Lambda \E_{C} C^{\ot t} \ketbra{\psi}{\psi} C^{\ot t,\dagger}] \geq 1 - O(t^2/d)\,,\]
where $\Lambda$ denotes the projector on the \emph{distinct subspace}.
In other words, after applying the random Clifford on the arbitrary input state $\ket{\psi}$, the resulting state has so little weight on the non-distinct subspace that we only need to deal with the distinct subspace.

For states $\ket{\psi}$ in the distinct subspace, we show that a $t$-wise $PF$ twirl results in
\begin{equation*}
    \E_{PF} (PF)^{\ot t}  \proj{\psi} ((PF)^{\ot t})^\dagger = \sum_{\lambda \vdash t} \frac{\ind_{\Lambda_\lambda}}{\Tr[\ind_{\Lambda_{\lambda}}]} \otimes \Tr_{W_\lambda}[\ind_{P_\lambda} \ketbra{\psi}{\psi}\ind_{P_\lambda}]\,,
\end{equation*}
where $\Lambda_\lambda$ is a subspace of $W_{\lambda}$ that comes from decomposing the distinct subspace projector $\Lambda$ in terms of the Schur-Weyl subspaces. 
We show that $\Lambda_{\lambda}$ fills most of $W_\lambda$, i.e.~the dimension of the subspace $\Lambda_\lambda$ is close to the dimension of $W_\lambda$:
$$\frac{\Tr[\ind_{\Lambda_{\lambda}}]}{\Tr[\id_{W_\lambda}]} = 1- O\left(\frac{t^2}{d}\right).$$
This implies that the mixed state on $\Lambda_\lambda$ is close in trace distance to the maximally mixed state on $W_{\lambda}$: 
$$\left\|\frac{\ind_{\Lambda_\lambda}}{\Tr[\ind_{\Lambda_{\lambda}}]} - \frac{\ind_{W_\lambda}}{\Tr[\ind_{W_\lambda}]}\right\|_1 = O \Big(\frac{t^2}d \Big)\,.$$

Since this is true for each $\lambda$, putting all the above together, we get that the left and right hand sides in \cref{eqn:inf-security} have trace distance\footnote{We lose a square-root factor in the analysis because of an application of the gentle measurement lemma.} at most $\sqrt{t}/d$, which is exponentially small since $t= \poly(n)$ and $d=2^n$.

\subsubsection{Pseudorandom isometries with adaptive security} \label{sec:intro_adaptive_pri}

Before explaining how we can construct adaptively secure PRIs, let us briefly discuss the obstacles to extending the proof of \Cref{thm:pfc-ensemble} to adaptive security.
If we were able to prove that the $PFC$ ensemble is a \emph{relative-error} design, adaptive security for our PRU construction would follow using known techniques (see e.g.~\cite[Section 3]{kretschmer2021quantum}). In particular, one can convert an adaptive algorithm that makes $t$-queries to a non-adaptive $t$-query algorithm with post-selection. Relative error designs allow one to obtain multiplicative approximations to the acceptance probabilities, allowing the whole argument to go through. 

However, we can only prove that the $PFC$ ensemble is a \emph{diamond-error} $t$-design, and  diamond error only gives additive approximations which are too weak to make this post-selection argument work. As outlined in first part of the proof overview, to prove that the $PFC$ ensemble is a superpolynomial design, we first use the random Clifford unitary to restrict our attention to the distinct-string subspace, and then analyse the $PF$ twirl on the distinct subspace. It turns out that a small modification of our proof of \cref{thm:pru_security} can show that the $PF$ twirl behaves like a (one-sided) \emph{relative-error} design \emph{on the distinct subspace} (\cref{lem:rel_error_distinct}). The problem is that the first step, the reduction to the distinct subspace, does not seem to work in the relative-error setting.
If one were able to perform the reduction to the distinct subspace for adaptive algorithms as conjectured in \cref{conj:scrambling}, adaptive PRU security would follow.

Fortunately the relative-error design property of the $PF$ ensemble on the distinct subspace is already useful by itself.
In the usual conversion from relative-error designs to adaptive security, one starts from the resource state $U^{\ot t} \proj{\Omega} U^{\ot t, \dagger}$, where $\ket{\Omega}$ is the maximally entangled state between $(\C^d)^{\ot t} \ot (\C^d)^{\ot t}$, and uses gate teleportation on this resource state to simulate adaptive queries to the unitary $U$.
This gate teleportation trick essentially reduces adaptive to non-adaptive security, since we now only need to analyse the application of $U$ in parallel on the resource state.
However, the gate teleportation trick introduces a large ``post-selection factor'', which can only be handled with relative-error designs.

Since we have a relative-error design (the $PF$ ensemble) on the distinct subspace, a natural idea is to perform the same gate teleportation trick, except using $U^{\ot t} \proj{\Omega_\Lambda} U^{\ot t, \dagger}$ as the resource state.
Here, $\ket{\Omega_\Lambda}$ is the unnormalised maximally entangled state on the distinct subspace of $(\C^d)^{\ot t}$, i.e.
\begin{align*}
\ket{\Omega_\Lambda} = \sum_{x_1, \dots, x_t \in [d] \text{ distinct}} \ket{x_1, \dots, x_t} \ket{x_1, \dots, x_t} \,.
\end{align*}

Of course $U^{\ot t} \proj{\Omega_\Lambda} U^{\ot t, \dagger}$ is not the ``correct'' resource state for performing gate teleportation; consequently, the resulting state is not simply the state with $U$ applied adaptively in the right places (which is what we get when using the correct resource state $U^{\ot t} \proj{\Omega} U^{\ot t, \dagger}$).
Instead, we get a slightly different state.

The challenge then is to show that the outputs of gate teleportation with the ``correct'' and ``distinct'' resource states are close.
This is the step that we are not able to prove for the general case of unitaries (although it seems very plausible that it can be proven using similar techniques).
However, if we fix $\omega(\log n)$ of the $n$ qubits of the input to the unitary, then we are able to show that the outputs of the gate teleportation for the two different resource states $U^{\ot t} \proj{\Omega} U^{\ot t, \dagger}$ and $U^{\ot t} \proj{\Omega_\Lambda} U^{\ot t, \dagger}$ are indeed close.
Fixing parts of the input to a unitary (to some universal state independent from the other input -- in our case the fixed input qubits are simply in a $\ket{+}$-state) turns the unitary into an isometry.
This is how we achieve PRIs with adaptive security.

Seeing how fixing part of the input to the unitary is helpful requires an explicit expression for the output state of the gate teleportation, so we defer this discussion to \cref{sec:pri_adaptive_security}.
We also remark that it seems likely that this proof strategy can be extended to adaptive PRUs: the only reason we need to relax our construction to PRIs is in order to ensure that the adaptive queries are in some adaptive version of the distinct ``subspace''.
If one could instead use query complexity arguments to prove this property for some unitary ensemble, one could simply replace our random Clifford unitary by this ensemble and achieve adaptively secure PRUs.
We formalise this idea in \cref{sec:towards-adaptive}.

\subsection{Discussion and future directions} \label{sec:discussion}

In this work we have taken a unified approach to constructing $t$-designs and PRUs by introducing the $PFC$ ensemble and showing that it mimicks the Haar measure extremely well.
Derandomising the $PFC$ ensemble with existing classical primitives such as $t$-wise independent functions or pseudorandom functions yields surprisingly simple constructions of diamond-error $t$-designs with circuit depth linear in $t$ and PRUs with non-adaptive security.

There are several ways in which one would like to improve upon these results and we hope that our techniques can be extended accordingly.
Our $t$-design construction only achieves small diamond error (corresponding to non-adaptive security) with linear scaling in $t$, but requires quadratic scaling in $t$ to achieve small relative error.
Constructing a relative-error $t$-design with linear scaling in $t$ remains an interesting open problem. 
It may well be that the $PFC$ ensemble is a relative-error $t$-design with linear scaling in $t$, but our current analysis does not show this because we do not know how to analyse the $PFC$ ensemble on the non-distinct subspace.
We can extend the analysis to a subspace with a constant number of collisions, but going beyond that seems to require new ideas.
One can also optimise the dependence in the number of qubits $n$; for this, we refer to~\cite{random_walks}.

Similarly, our PRU construction only achieves non-adaptive security, and in contrast to $t$-designs we cannot ``amplify'' the PRU construction to relative/adaptive security because that amplification requires $t$ to be known ahead of time. Looking forward, there are three directions in which one would like to extend the security guarantee of our PRU construction: allowing adaptive queries, allowing inverse queries, and allowing controlled queries to the unitary.
Our construction plausibly has adaptive  security and we discuss an approach to extending our current proof to the adaptive setting in \cref{sec:towards-adaptive}.
In contrast, the $PFC$ ensemble is \emph{not} indistinguishable from Haar random unitaries with inverse queries: this is because applying $C^\dagger F^\dagger P^\dagger$ to $\ket{0}$ always yields a stabiliser state, but applying a Haar random unitary to $\ket{0}$ does not, and this difference can be tested efficiently.
However, if one simply adds another independent Clifford at the end (i.e.~considers $C'PFC$), the construction is plausibly secure against inverse queries, but we do not know how to analyse this.
Proving security with controlled access to the unitary is similarly unclear.

Aside from constructing PRUs (or PRIs), their applications are also largely unexplored. Given the utility of PRFs in classical computer science, PRUs appear to be a fundamental primitive for quantum computer science, but relatively few concrete applications have been proposed. Below we briefly discuss some applications that have been mentioned in the literature and suggest some new ones.

One natural area of application is quantum cryptography. 
For example,~\cite{lu2023quantum,ananth2023pseudorandom} showed that PRSSs and PRIs can be useful for multi-copy quantum cryptography:
in most encryption schemes for quantum messages (e.g., the quantum one-time pad and its variants), if we wish to encrypt multiple identical copies of the same quantum state, we need to sample fresh keys for each new ciphertext.~\cite{lu2023quantum,ananth2023pseudorandom} observed that Haar (pseudo)randomness allows one to perform such a task in a compact manner with only a single key (for an arbitrary polynomial amount of identical copies). If the state to be encrypted is guaranteed to be unentangled with the environment, the PRSSs and PRIs from~\cite{lu2023quantum,ananth2023pseudorandom} suffice; in the general case which allows for entanglement with an auxiliary system, non-adaptively PRUs as constructed in our work seem necessary. We remark that \cite{ananth2023pseudorandom} also use PRIs (rather than PRUs) to construct succinct quantum state commitments, many-copy unforgeable quantum message authentication schemes, as well as to increase the length of pseudorandom quantum states.
As another example, PRUs might be useful in the context of unclonable cryptography. Many constructions, such as those for unclonable encryption~\cite{broadbent2020uncloneable} or quantum copy-protection~\cite{coladangelo2020quantum,coladangelo2021hidden}, make use of either Wiesner states or subspace coset states---both of which are completely broken once identical copies become available. It seems plausible that one could use PRUs to construct multi-copy secure unclonable encryption schemes, and even multi-copy secure quantum copy-protection schemes. 

Another area of application concerns the time evolution of chaotic quantum systems. In the past few years, a series of works has proposed using Haar random unitaries as ``perfect scramblers''~\cite{Hayden_2007,brown2012scrambling,susskind2016computational} to model such dynamics. However, a more recent line of work has instead shifted towards quantum pseudorandomness in order to model such phenomena in terms of \emph{efficient} processes.
For example, Kim and Preskill~\cite{Kim_2023} use PRUs to model the internal dynamics of a black hole, whereas Engelhardt et al.~\cite{engelhardt2024cryptographic} use PRUs to model the time evolution operator of a holographic conformal field theory. Due to their scrambling properties, it is conceivable that PRUs will find more applications in theoretical physics.

\paragraph{Related independent work.}
After the initial announcement of our results on pseudorandom unitaries on the arXiv~\cite{metger2024pseudorandom} (which this present paper supersedes), we were made aware of independent work by Chen, Bouland, Brandao, Docter,  Hayden, and Xu~\cite{chen2024prus}, who achieve similar results using a different ensemble of unitaries. In another related independent work, Brakerski and Magrafta~\cite{brakerski2024real} construct real-valued unitaries that look Haar random on any polynomial-sized set of orthogonal input states; they use an ensemble that consists of the product of a Hadamard, a random permutation, and a binary phase operator.

\paragraph{Acknowledgments.} We thank Prabhanjan Ananth, Adam Bouland, Chi-Fang Chen, Tudor Giurgica-Tiron, Jonas Haferkamp, Robert Huang, Isaac Kim, Fermi Ma, Xinyu Tan, and John Wright for helpful discussions. We also thank John for suggesting a simplified proof of \Cref{lem:K_simplified} and Jonas for pointing out a proof of \Cref{lem:diamond_amp}.
We thank the Simons Institute for the Theory of Computing, where some of this work was conducted.
TM acknowledges support from the ETH Zurich Quantum Center, the SNSF QuantERA project (grant 20QT21\_187724), the AFOSR grant FA9550-19-1-0202, and an ETH Doc.Mobility Fellowship. AP is supported by the National Science Foundation (NSF) under Grant No. CCF-1729369. MS acknowledges support from the NSF award QCIS-FF: Quantum Computing \& Information Science Faculty Fellow at the University of Illinois Urbana-Champaign (NSF 1955032). 
HY is supported by AFOSR awards
FA9550-21-1-0040 and FA9550-23-1-0363, NSF CAREER award CCF-2144219, NSF award CCF-2329939, and the Sloan Foundation.

\section{Preliminaries}

\subsection{Notation and basic definitions}
\label{sec:notation}

For $N\in \N$, we use $[N] = \{1,2,\dots,N\}$ to denote the set of integers up to $N$. We oftentimes identify elements $x \in [N]$ with bit strings $x \in \bit^n$ via their binary representation whenever $N=2^n$ and $n \in \N$. We use $\delta_{i,j}$ to denote the delta function which is $1$ iff $i=j$ and 0 otherwise. The notation $x \sim X$, for a set $X$, describes that an element $x$ is drawn uniformly at random from $X$. Similarly, if $\algo D$ is a distribution, we let
$x \sim \algo D$ denote that $x$ is sampled according to $\algo D$. The statistical distance between two distributions $\mathcal{D}$ and $\mathcal{D}'$ over a set $X$ is defined as $\| \mathcal{D} - \mathcal{D}'\|_1 = \sum_{x \in X} |\mathcal{D}(x)-  \mathcal{D}'(x)|$.

\paragraph{Linear maps.} For a Hilbert space $\cH$, we denote by $\linear(\cH)$ linear operators on $\cH$.
A map $\cM: \linear(\cH) \to \linear(\cH)$ is called a quantum channel if it is completely positive and trace-preserving.
If $\cH'$ is an additional Hilbert space and $O \in \linear(\cH \ot \cH')$ is an operator on a larger space, we write $\cM(O)$ to mean $\cM$ applied to the $\cH$-subsystem, i.e.~$\cM(O)$ is shorthand for $(\cM \ot \id_{\cH'})(O)$ where $\id_{\cH'}: \linear(\cH') \to \linear(\cH')$ denotes the identity map. For a subspace $W$ of a vector space $V$, we also denote by $\ind_{W}$ the orthogonal projection onto $W$ (i.e. the identity on the subspace $W$). We use the notation $\mathrm{U}(d)$ and $\mathrm{C}(d)$ to denote the $d$-dimensional unitary group and Clifford group, respectively.

\paragraph{Distance measures.}
For a linear operator $A \in \linear(\cH)$, we use $\norm{A}_1 = \Tr[(A^\dagger A)^{1/2}]$ to denote the Schatten 1-norm (also called trace norm) and $\norm{A}_\infty$ to denote the Schatten $\infty$-norm (also called operator norm), which is defined to be the largest singular value of $A$.
The trace distance between two operators $A, B \in \linear(\cH)$ is defined as $\norm{A - B}_1$.

For two quantum channels $\cM, \cN: \linear(\cH) \to \linear(\cH')$, we define the diamond distance
\begin{align*}
\norm{\cM - \cN}_{\Diamond} = \max_{\ket{\psi} \in \cH \ot \cH} \norm{(\cM \ot \id_{\cH})(\proj{\psi}) - (\cN \ot \id_{\cH})(\proj{\psi})}_1 \,,
\end{align*}
where $\id_{\cH}: \linear(\cH) \to \linear(\cH)$ denotes the identity channel.
The diamond norm is submultiplicative, i.e.~for two non necessarily completely positive superoperators $\cE, \cF$ (e.g.~$\cE = \cF = \cM - \cN$ could be the difference between channels $\cM$ and $\cN$), it holds that
\begin{align}
\norm{\cE \circ \cF}_{\Diamond} \leq \norm{\cE}_{\Diamond} \cdot \norm{\cF}_{\Diamond}\,. \label{eqn:diamond_submult}
\end{align}

\paragraph{Distinct tuples.} We use bold faced fonts to denote tuples. For a tuple $\bx = (x_1,\dots,x_t) \in [d]^t$ and a permutation $\sigma \in S_t$, we write $\bx_{\sigma} = (x_{\sigma(1)}, \cdots,x_{\sigma(t)})$ for the tuple where the indices are permuted according to $\sigma.$ We call a tuple $\bx \in [d]^t$ distinct if $x_i \neq x_j$ for all $i \neq j$. We denote the set of distinct tuples in $[d]^t$ by $\distinct(d,t)$.
We also define the projector onto the subspace of distinct tuples 
\begin{equation}\label{def:distinct}
    \ \Lambda_{d,t} = \sum_{\substack{\bx \in \distinct(d,t)}} \proj{\bx} \,.
\end{equation}
We will frequently drop the indices $d,t$ and just write $\Lambda$, with $d$ and $t$ being clear from context.

In the proof of adaptive security for our PRI construction, we will also make use of the maximally entangled state between two copies of the distict subspace, which we denote by 
\begin{align*}
\ket{\Omega_{\Lambda_{d,t}}} = \frac{1}{\sqrt{|\distinct(d,t)|}} \sum_{\bx \in \distinct(d,t)} \ket{\bx} \ket{\bx} \,.
\end{align*}
As for the projector onto the distinct string subspace, we will frequently drop the indices $d, t$ and simply write $\ket{\Omega_\Lambda}$

\paragraph{Permutation  operators.} \label{sec:prelim_phase_perm}

We define the following permutation operator on $\C^d$.

\begin{definition}[Permutation operator on $\C^d$]
Define the permutation operator $P_\pi$ on $\C^d$ for $\pi \in S_d$ to be the linear map
\begin{align}\label{def:permutationop}
P_\pi: \ket{x} \mapsto \ket{\pi(x)}\,.
\end{align}
\end{definition}

We will frequently consider uniformly random permutation operators on $\C^d$. We will suppress the dependence on $\pi$ and write the random operator as $P$.

\paragraph{Symmetric group and representations.}

Unitary representations of a group allow us to represent the elements of the group as unitary matrices over a vector space in a way that the group operation is represented by matrix multiplication. We consider the following representation of the symmetric group which permutes the tensor factors.

\begin{lemma}[Representation of $S_t$ on tensor product spaces]
For any fixed $d$, define the permutation operator $R_\pi$ on $(\C^d)^{\ot t}$ for $\pi \in S_t$ to be the map
\[
    R_\pi: \ket{\ba} \mapsto \ket{\ba_{\pi^{-1}}}~.
\]
Then $((\C^d)^{\ot t}, R_\pi$ forms a unitary representation of $S_t$.
Note that we leave the dependence of $R_\pi$ on the choice of $d$ implicit.
\end{lemma}

We note that the representation $((\C^d)^{\ot t}, R_{(\cdot)})$ can be decomposed into (isotypic) copies of the irreducible representations (or irreps) of the symmetric group $S_t$, which we denote by $\{(V_{\lambda}, R^{\lambda}_{(\cdot)})\}_{\lambda}$, where $\lambda \vdash t$ is a partition of $t$ (or Young diagram with at most $t$ boxes) and $V_\lambda$ are vector spaces called Specht modules.  

\paragraph{Binary phase operators.} For a function $f: [d] \to \bit$, we define the binary phase operator 
\begin{align}\label{def:binaryphaseop}
F_f: \ket{x} \mapsto (-1)^{f(x)} \ket{x}.
\end{align} 
We will frequently consider uniformly random binary phase operators, which we will just write as $F$.

\subsection{Haar measure and unitary designs}

We recall the definition of the Haar measure and $t$-wise twirl.

\begin{definition}[Haar measure]
The Haar measure is the unique left- and right-invariant probability measure on the unitary group $\unitary(d)$.
Throughout this paper, we denote sampling from the Haar measure over $\unitary(d)$ by $U \sim \HaarMeasure(d)$. If the dimension $d$ is clear from the context, we simply write $U \sim \HaarMeasure$.
\end{definition}

\begin{definition}[$t$-wise twirl] \label{def:twirl}
Let $\nu$ be an ensemble of unitary operators in $\mathrm{U}(d)$. Then, the $t$-wise twirl (also called $t$-th moment operator) with respect to $\nu$ is defined as the operator
\begin{align*}
\mathcal{M}_\nu^{(t)}(\cdot) = \E_{R\sim \nu} R^{\ot t} (\cdot) R^{\ot t, \dagger} \,.
\end{align*}
\end{definition}
If $R$ is a random unitary matrix whose distribution is clear in context, we oftentimes use the shorthand notation $\momr{\cdot}$. For example, we frequently consider the $t$-wise $\mompft{\cdot}$ twirl, where $R = PF$ is a product of a random permutation operator $P$ and a binary phase operator $F$, which are sampled independently.
We call this the permutation-phase twirl.
Similarly, we use $\mompfct{\cdot}$ to denote the ``$PFC$-twirl'', where $R = P F C$ for $P$ and $F$ as before and $C$ a uniformly random Clifford unitary. Finally, we denote by $\mathcal{M}_{\Haar}^{(t)}(\cdot)$ the $t$-wise twirl over unitaries which are sampled according to the Haar measure.

Just like $t$-wise independent functions ``simulate'' uniformly random functions if one only considers $t$-th moments, there is a notion of simulating the Haar measure up to the $t$-th moment, called a \emph{$t$-design}.
Various notions of approximation exist for $t$-designs and we refer to \cite{mele2023introduction} for an overview.
Here, we will use the following two notions of approximate $t$-designs.
\begin{definition}[Diamond $\eps$-approximate $t$-design] \label{def:additive_error_design}
An ensemble $\nu$ of unitary operators in $\mathrm{U}(d)$ is called a diamond $\eps$-approximate $t$-design, if
$$
\left\| \mathcal{M}_{\nu}^{(t)} - \mathcal{M}_{\Haar}^{(t)} \right\|_\Diamond \leq \eps.
$$  
\end{definition}

\begin{definition}[Relative-error $\eps$-approximate $t$-design] \label{def:relative-error-design}
An ensemble $\nu$ of unitaries acting on a Hilbert space $\cH$ is a relative-error $\eps$-approximate $t$-design if, for all states $\ket{\psi} \in \cH \ot \cH$, the following operator inequality holds
\begin{align*}
(1- \eps) \mathcal{M}_{\nu}^{(t)}(\proj{\psi}) \leq \momhaart{\proj\psi} \leq (1+ \eps) \mathcal{M}_{\nu}^{(t)}(\proj{\psi}) \,.
\end{align*}
We call a random unitary a left-one-sided or right-one-sided relative-error $\eps$-approximate $t$-design if only the left or right side of these inequalities holds.
\end{definition}

Every diamond error $t$-design is also a relative-error $t$-design, but this conversion incurs a large blowup in the error.
\begin{lemma}[{\cite[Lemma 3]{brandao2016local}}] \label{lem:diamond-to-rel}
Suppose a random unitary $R$ in $d$ dimensions is a diamond $\eps$-approximate $t$-design.
Then $R$ is also a relative-error $(\eps \cdot d^{2t})$-approximate $t$-design.
\end{lemma}

\subsection{Schur-Weyl duality}

Consider the representation $R_\pi$  of the symmetric group $S_t$ and the representation $U^{\ot t}$ of the unitary group $\mathrm{U}(d)$ both on the vector space $(\C^d)^{\ot t}$. Schur-Weyl duality says that the irreducible subrepresentations of both these representations fit together nicely. 

\begin{lemma}[{Schur-Weyl duality, see e.g.~\cite[Theorem 1.10]{christandl2006structure}}] \label{lem:schur-weyl}
The tensor product space $(\C^d)^{\ot t}$ can be decomposed as
\[
    (\C^d)^{\ot t} \cong \bigoplus_{\lambda \vdash t} P_{\lambda}~ \text{ with } P_{\lambda} = W_{\lambda} \ot V_{\lambda}
\]
where $\lambda \vdash t$ indexes partitions of $\{1, \dots, t\}$, which are commonly represented by Young diagrams.\footnote{We note that throughout this paper $d \gg t$, otherwise the Young diagrams need to be restricted to $d$ rows.}

The \emph{Weyl modules} $W_\lambda$ are irreducible subspaces for the unitary group $\mathrm{U}(d)$ and the \emph{Specht modules} $V_\lambda$ are irreducible subspaces for the symmetric group $S_t$. Consequently, the action of the product group $S_t \times U_d$ on $(\C^d)^{\ot t}$ decomposes as
\begin{align*}
    R_\pi  = \sum_{\lambda \vdash t} \ind_{W_{\lambda}} \ot R_\pi^{(\lambda)} \text{\hphantom{t} and \hphantom{t}} &  U^{\ot t}  = \sum_{\lambda \vdash t} U^{(\lambda)} \ot \ind_{V_{\lambda}}~,
\end{align*}
where $(W_\lambda, U^{(\lambda)})$ and $(V_\lambda, R_\pi^{(\lambda)})$ are irreducible representations of the unitary group $\unitary(d)$ and the symmetric group $S_t$, respectively. 
\end{lemma}

We denote the specific basis which block-diagonalizes all the above operations as the Schur-Weyl basis.

\begin{definition}[Schur-Weyl basis]
Let $\{\ket{w_{\lambda, i}}\}_{i}$ and $\{\ket{v_{\lambda, j}}\}_j$ be orthonormal bases of $W_\lambda$ and $V_\lambda$, respectively.
Then we call 
\begin{align*}
\{\ket{w_{\lambda, i}} \ot \ket{v_{\lambda, j}}\}_{\lambda, i, j}
\end{align*}
a Schur-Weyl basis of $(\C^d)^{\ot t}$ (where we interpret each vector $\ket{w_{\lambda, i}} \ot \ket{v_{\lambda, j}} \in P_\lambda$ as a $(\C^d)^{\ot t}$-vector in the natural way).
\end{definition}

The following decomposition of the distinct subspace projector $\Lambda$ easily follows from Schur-Weyl duality: since $\Lambda$ is invariant under permutation of the tensor factors (i.e.~$R_{\pi} \Lambda = \Lambda R_{\pi}$ for all $\pi \in S_t$), it acts as an identity on the Specht modules $V_{\lambda}$ by Schur's lemma. The same decomposition also holds for any permutation-invariant operator but we shall only need the following.

\begin{lemma}[Decomposition of the distinct subspace projector] \label{lem:perm_inv_sw}
Let $\Lambda \in \linear((\C^d)^{\ot t})$ be the projector on the tuples of distinct strings defined in \Cref{def:distinct}. 
Then, 
\begin{align*}
\Lambda = \sum_{\lambda \vdash t} \Lambda^{(\lambda)}_{W_\lambda} \ot \id_{V_\lambda} \,,
\end{align*}
with each $\Lambda^{(\lambda)}_{W_{\lambda}}$ a projector on a subspace of $W_\lambda$. 
\end{lemma}

We will also need the following relation between the dimensions of the Weyl and Specht modules.
\begin{lemma}[\cite{christandl2006structure}, Theorem 1.16] \label{lem:dim_weyl_specht}
The dimensions of the Weyl and Specht modules $W_\lambda$ and $V_\lambda$ satisfy
\[
    \dim(W_\lambda) = \frac{\dim(V_\lambda)}{t!} \prod_{(i,j) \in \lambda} (d + j - i) \,,
\]
where $(i,j)$ denotes the row and column number of a box in the Young diagram corresponding to $\lambda$. 
\end{lemma}

The following lemma follows from the standard formula for computing the projector on isotypical copies of an irreducible subspace (see~\cite[Section 2.4]{fulton2013representation}). Here we apply it to the representation of the symmetric group $S_t$  over $(\C^d)^{\ot t}$, where Schur-Weyl duality implies that the subspace of all isotypic copies of the Specht modules $V_{\lambda}$ is exactly $P_\lambda = W_\lambda \ot V_\lambda$.

\begin{lemma}\label{lem:P_characters}
The projection onto the subspace $P_\lambda$ is given by 
\begin{equation}
\label{eq:isotypical-projection}
    \ind_{P_\lambda} = \frac{\dim(V_\lambda)}{t!} \sum_{\pi \in S_t} \chi_\lambda(\pi^{-1}) R_\pi \,,
\end{equation}
where $\chi_\lambda(\cdot) = \Tr[R^{\lambda}_{(\cdot)}]$ is the character corresponding to the irrep $(V_{\lambda}, R^{(\lambda)}_{(\cdot)})$.
\end{lemma}

We will also need (a special case of) the standard Schur orthogonality relations for matrix coefficients (see~\cite{Bump}, Theorems 2.3 and 2.4). These relations say that if we express unitary irreducible representations of a group in any basis, then the different matrix entries are orthogonal under an inner product obtained by averaging over the group. Here we specialize the above to the Schur-Weyl basis and the unitary irreducible subrepresentations $R^{\lambda}_{(\cdot)}$ of $R_{(\cdot)}$ as given in \Cref{lem:schur-weyl}.
\begin{lemma}[Schur orthogonality relations] \label{lem:orthogonality}
    Let $(V_{\lambda}, R^{\lambda}_{(\cdot)}), (V_{\lambda'}, R^{\lambda'}_{(\cdot)})$ be two irreducible representations of the symmetric group $S_t$. Then,
    \[ \E_{\pi \in S_t}\left[ \bra{v_{\lambda, i}} R^{\lambda}_{\pi}  \ket{v_{\lambda, j}}  \overline{\bra{v_{\lambda', k}} R^{\lambda'}_{\pi}  \ket{v_{\lambda', \ell}}} \right] = \frac{1}{\dim(V_{\lambda})} \delta_{\lambda, \lambda'} \delta_{i,k} \delta_{j,\ell}\,. \]
\end{lemma}

\section{Analysis of the $PFC$ ensemble}
\label{sec:pru_security}

Our main technical result is to introduce a new ensemble of random unitaries, the concatenation of a random Clifford, a random binary phase, and a random basis permutation, and show that this ensemble is a $t$-design even for superpolynomially large $t$.
More formally, we prove the following.

\begin{theorem}\label{thm:pfc-ensemble}
Let $t \in \N$. For $d = 2^n$, the $d$-dimensional $PFC$ random unitary ensemble (see \Cref{def:pfc-ensemble}) is a diamond $\eps$-approximate $t$-design for $\eps = O(t/\sqrt{d})$.
\end{theorem}

\begin{proof} 
Let $P$ denote a uniformly random permutation operator on $n$-qubits as defined in \Cref{def:permutationop}, let $F$ be a uniformly random binary phase operator as in \Cref{def:binaryphaseop} and let $C$ be a uniformly random $n$-qubit Clifford, all sampled independently. 
Let $\reg{A}$ be an $n$-qubit quantum register, $\reg E$ another quantum register isomorphic to $A$, and consider an arbitrary state $\ket{\psi}_{\reg{AE}}$.
By \Cref{def:additive_error_design}, it suffices to argue that the density matrices 
\begin{align*}
\rho \deq \E_{P, F, C} (P F C)^{\ot t}_{\reg A} \proj{\psi}_{\reg{AE}} (P F C)^{\ot t, \dagger}_{\reg A} 
\quad \text{and} \quad \rho^{\rm hr}: = \E_{U \sim \Haar} U^{\ot t}_{\reg A} \proj{\psi}_{\reg{AE}} U^{\ot t, \dagger}_{\reg A}
\end{align*}
are $O(t/\sqrt{d})$-close in trace distance.
It will be convenient to define 
\begin{align*}
\xi_{\reg{AE}} \deq \E_{C} C^{\ot t}_{\reg A} \proj{\psi}_{\reg{AE}} C^{\ot t, \dagger}_{\reg A} \,.
\end{align*}
Recalling the definition of $t$-wise $R$-twirl operator (acting on the register $\reg{A}$, which we shall omit from the notation henceforth), our goal is to bound the following trace distance:
\begin{align}\label{eqn:td}
    \ \|\rho - \rho^{\rm hr}\|_1 = \left\|\mompft{\xi_{\reg{AE}}} - \momhaart{ \xi_{\reg{AE}}}\right\|_1 \,,
\end{align}
where we used the fact that the product unitary $UC$ is also Haar distributed by the invariance of the Haar measure, so $\momhaart{ \xi_{\reg{AE}}} = \momhaart{\proj{\psi}_{\reg{AE}}}$.

\paragraph{Step 1: Reduction to distinct subspace.} To show the above, the calculations are easier if we restrict attention to the subspace of $(\C^d)^{\ot t}$ that consists of distinct strings, i.e.~basis states of the form $\ket{\bx}$ where $\bx = (x_1,\dots,x_t) \in [d]^t$ is a tuple of distinct strings. We show in \Cref{sec:clifford} that applying a $t$-wise Clifford twirl ensures that the input state has a large overlap with this subspace. 

\begin{lemma}[Clifford twirl and distinct subspace] \label{lem:clifford}
Let $\Lambda$ be the projector on the distinct subspace defined in \Cref{def:distinct}. Then, for any state $\rho$ on the register $\reg{A}$, we have
    $$\Tr[\Lambda \E_{C} {{C}^{\ot t}} \rho \, {{C}^{\ot t,\dagger}}] \geq 1 - O(t^2/d)\,.$$
\end{lemma}

Let $\phi_{\reg{AE}}$ be the mixed state obtained by normalizing the (positive semi-definite) matrix $\Lambda_{\reg{A}} \xi_{\reg{AE}} \Lambda_{\reg{A}}$, where $\xi_{\reg{AE}}$ is defined in \Cref{eqn:td}. Then the above together with the Gentle Measurement lemma (see \cite[Lemma 9]{Winter_1999}) implies that $\|\phi - \xi\|_1 \le O(t/\sqrt{d})$. Consider a purification $\ket{\phi}_{\reg{A \tilde E}}$ of the state $\phi_{\reg{AE}}$ which satisfies $\Lambda_{\reg A} \ket{\phi}_{\reg{A \tilde E}} = \ket{\phi}_{\reg{A \tilde E}}$; such a purification exists by construction of $\phi_{\reg{AE}}$. Then, 
\begin{align*}
\norm{\rho - \rho^{\rm hr}}_1 
&\leq \norm{\mompft{\phi_{\reg{A E}}} - \momhaart{\phi_{\reg{A E}}}}_1 + O(t/\sqrt{d})\\
&\leq \norm{\mompft{\proj{\phi}_{\reg{A \tilde E}}} - \momhaart{\proj{\phi}_{\reg{A \tilde E}}}}_1 + O(t/\sqrt{d}) \,, \numberthis \label{eqn:haarvspf}
\end{align*}
where we use that a $t$-wise twirl is a quantum channel and the 1-norm can only decrease under partial trace. 
Thus, we may assume for the rest of the proof that the input is supported over the distinct subspace of register $\reg{A}$.

In order to bound \Cref{eqn:haarvspf}, we find explicit expressions for both states $\mompft{\proj{\phi}_{\reg{A \tilde E}}}$ and $\momhaart{\proj{\phi}_{\reg{A \tilde E}}}$ in the Schur-Weyl basis.

\paragraph{Step 2: Action of Haar twirl.}
Recall that we consider states $\ket{\phi}_{\reg{A \tilde E}}$ where the register $\reg{A}$ is supported over the distinct subspace, i.e.~states satisfying $\Lambda_{\reg A} \ket{\phi}_{\reg{A\tilde E}} = \ket{\phi}_{\reg{A\tilde E}}$.
To apply Schur-Weyl duality, we decompose $\reg{A} \cong (\mathbb{C}^d)^{\ot t} \cong \bigoplus_{\lambda \vdash t} P_\lambda$ (where $P_\lambda = W_{\lambda} \ot V_{\lambda}$) in terms of the Weyl and Specht modules $W_{\lambda}$ and $V_{\lambda}$, respectively.
We first analyse the output of the Haar twirl in the Schur-Weyl basis, proving the following decomposition.

\begin{lemma}[Action of Haar twirl]\label{lem:haar-twirl} Let $\rho_{\lambda} = \dfrac{\ind_{W_\lambda}}{\Tr[\ind_{W_\lambda}]} $ be the maximally mixed state on $W_{\lambda}$.
Then
\begin{align}\label{eqn:haar-twirl}
    \     \momhaart{\proj{\phi}_{\reg{A \tilde E}}} = \sum_{\lambda \vdash t} \rho_{\lambda} \otimes \Tr_{W_\lambda}[\ind_{P_\lambda} \proj{\phi}_{\reg{A\tilde E}}\ind_{P_\lambda}]\,. 
\end{align}
\end{lemma}

We prove the above lemma in \Cref{sec:haar-twirl}. In fact, the same proof also shows that \Cref{lem:haar-twirl} holds for arbitrary input states, not just input states supported only on the distinct string subspace. However, we will not need this since we are only interested in bounding the trace distance in \Cref{eqn:haarvspf}, which considers states on the distinct string subspace.

\paragraph{Step 3: Action of permutation-phase twirl on distinct subspace.}
Next, recalling the decomposition of the distinct subspace projector given by \Cref{lem:perm_inv_sw}, we show that the action of the permutation-phase twirl results in the following; we defer the proof to \Cref{sec:pf-twirl}.

\begin{lemma}[Action of permutation-phase twirl]\label{lem:pf-twirl} 
Denoting by $\sigma_{\lambda} = \dfrac{\Lambda^{(\lambda)}_{W_\lambda}}{\Tr[\Lambda^{(\lambda)}_{W_\lambda}]} $ the maximally mixed state on the subspace $\supp(\Lambda^{(\lambda)}_{W_\lambda}) \cap W_{\lambda}$, we have 
    \[ \mompft{\proj{\phi}_{\reg{A \tilde E}}} = \sum_{\lambda \vdash t} \sigma_{\lambda} \otimes \Tr_{W_\lambda}[\ind_{P_\lambda} \proj{\phi}_{\reg{A \tilde E}}\ind_{P_\lambda}]\,.\]
\end{lemma} 

\paragraph{Step 4: Haar twirl vs permutation-phase twirl.} 
We now have explicit expressions for both of the twirling operations in \Cref{eqn:haarvspf}.
We will use these expressions to show that 
\begin{align}
\norm{\momhaart{\proj{\phi}_{\reg{A\tilde E}}} - \mompft{\proj{\phi}_{\reg{A \tilde E}}}}_1 \leq O(t^2/d) \,. \label{eqn:final_comparison}
\end{align}
Plugging this into \Cref{eqn:haarvspf} completes the proof of the theorem.

To prove \Cref{eqn:final_comparison}, first note that all the subspaces $P_{\lambda}$ are orthogonal.
Therefore, inserting the expressions from \Cref{lem:haar-twirl} and \Cref{lem:pf-twirl} we have 
\begin{align}\label{eqn:main1}
\norm{\momhaart{\proj{\phi}_{\reg{A \tilde E}}} - \mompft{\proj{\phi}_{\reg{A \tilde E}}}}_1 & = \sum_{\lambda} \norm{\left( \rho_{\lambda} - \sigma_{\lambda} \right) \ot \Tr_{W_\lambda}[\ind_{P_\lambda} \proj{\phi}_{\reg{A\tilde E}}\ind_{P_\lambda}]}_1 \notag \\
&= \sum_{\lambda} \norm{ \rho_{\lambda} - \sigma_{\lambda}}_1 \cdot \norm{\Tr_{W_\lambda}[\ind_{P_\lambda} \proj{\phi}_{\reg{A\tilde E}}\ind_{P_\lambda}]}_1 \notag \\
\ & \le \max_{\lambda} \norm{ \rho_{\lambda} - \sigma_{\lambda}}_1,
\end{align}
where the second equality used that the 1-norm (trace norm) is multiplicative under tensor products, while the inequality follows from the fact that $\sum_{\lambda} \|\Tr_{W_\lambda}[\ind_{P_\lambda} \proj{\phi}_{\reg{A\tilde E}}\ind_{P_\lambda}]\|=1$, which can be seen by taking the trace on both sides of \Cref{eqn:haar-twirl}.

Note that in general, given a vector space $A = B \oplus B^\bot$, we have that 
\[ \norm{ \frac{\id_{A}}{\dim(A)} - \frac{\id_{B}}{\dim(B)}}_1 = \norm{\frac{\id_{B}}{\dim(A)} - \frac{\id_{B}}{\dim(B)}}_1 + \norm{\frac{\id_{B^\bot}}{\dim(A)}}_1 = 2-2\frac{\dim(B)}{\dim(A)}\,.\]
Since $\Lambda^{(\lambda)}_{W_\lambda}$ is a projector on a subspace of $W_{\lambda}$, we obtain the following for any $\lambda \vdash t$:
\begin{align}\label{eqn:main2}
    \norm{ \rho_{\lambda} - \sigma_{\lambda}}_1 &= \norm{ \frac{\id_{W_{\lambda}}}{\Tr[\id_{W_{\lambda}}]} - \frac{\Lambda^{(\lambda)}_{W_\lambda}}{\Tr[\Lambda^{(\lambda)}_{W_\lambda}]}}_1 = 2 - 2 \frac{\Tr[\Lambda^{(\lambda)}_{W_\lambda}]}{\Tr[ \id_{W_{\lambda}}]}\,.
\end{align}

We first compute the trace in the numerator of the right hand side.

\begin{claim}\label{clm:distinct-trace}
${\Tr[\Lambda^{(\lambda)}_{W_\lambda}]} = \dfrac{\dim(V_{\lambda})}{t!} \cdot{\Tr[\Lambda]}$.
\end{claim}
\begin{proof}
   \Cref{lem:perm_inv_sw} implies that $\Lambda = \sum_{\lambda \vdash t} \Lambda^{(\lambda)}_{W_{\lambda}} \otimes \id_{V_{\lambda}}$. Since $P_{\lambda} = W_{\lambda} \ot V_{\lambda}$, it follows that
    \begin{align}\label{eqn:dis-trace}
   \Tr[\Lambda^{(\lambda)}_{W_\lambda}] = \frac{1}{\dim V_\lambda} \Tr[\Lambda \id_{P_\lambda}]\,.
    \end{align}
Plugging in the expression for the projector $\id_{P_{\lambda}}$ from \Cref{lem:P_characters}, we get that
    \begin{align*}
    \Tr[\Lambda \id_{P_\lambda}] = \frac{\dim(V_\lambda)}{t!} \sum_\pi \chi_\lambda(\pi^{-1}) \Tr[\Lambda R_\pi] = \frac{\dim(V_\lambda)^2}{t!} \Tr[\Lambda] \,, \label{eqn:lambda_P_ip}
    \end{align*}
    where we used the fact that $\Tr(R_\pi \Lambda) = 0$ unless $\pi = e$, and that $\chi_\lambda(e) = \Tr[\id_{V_{\lambda}}] = \dim(V_\lambda)$. 
    This yields the desired result after insertion into \Cref{eqn:dis-trace}.
\end{proof}

Given the above, we can now compute the quantity in \Cref{eqn:main2} by using the dimension bounds for Weyl and Specht modules.

\begin{claim}\label{cm:ratio} $1- \dfrac{\Tr[\Lambda^{(\lambda)}_{W_\lambda}]}{\Tr[ \id_{W_{\lambda}}]} \le O(t^2/d).$
\end{claim}
\begin{proof}
Using \Cref{clm:distinct-trace} and \Cref{lem:dim_weyl_specht}, we have  
\[ 1- \dfrac{\Tr[\Lambda^{(\lambda)}_{W_\lambda}]}{\Tr[ \id_{W_{\lambda}}]} = 1- \frac{\dim(V_{\lambda})\Tr[\Lambda]}{t! \dim(W_{\lambda})} = 1- \frac{\Tr[\Lambda]}{\Pi_{(i,j) \in \lambda} (d + j - i)}\,.\]

Note that $\Tr[\Lambda] = \frac{d!}{(d-t)!} \ge (d-t)^t$ and $\Pi_{(i,j) \in \lambda} (d + j - i) \le (d+t)^t$, since there are at most $t$ boxes in the Young diagram corresponding to $\lambda$, and the coordinates $i,j$ range from $1$ to $t$. Thus (assuming $t^2 < d$, as otherwise the claim becomes trivial),
\[ 1- \dfrac{\Tr[\Lambda^{(\lambda)}_{W_\lambda}]}{\Tr[ \id_{W_{\lambda}}]} \le 1-  \left(\frac{d-t}{d+t}\right)^t \le 1 - \Big(\frac{1 - t/d}{1+t/d} \Big)^t \leq O(t^2/d)\,.\]
\end{proof}

Plugging the bound from \Cref{cm:ratio} into \Cref{eqn:main2} shows \Cref{eqn:final_comparison}.
This completes the proof of \Cref{thm:pfc-ensemble}.
\end{proof}

\subsection{Clifford twirl and distinct subspace (proof of \Cref{lem:clifford})}
\label{sec:clifford}

We will omit the registers $\reg{A} = \reg{A_1 \dots A_t}$ from the notation unless needed. Our goal is to show that applying a $t$-wise Clifford twirl on any input state $\rho$ produces a state that has a large overlap with the distinct subspace. The proof only relies on the standard fact that the Clifford group forms a $2$-design. We consider the projector on the orthogonal complement of the distinct subspace
    \begin{align*}
    \bar \Lambda = \id - \Lambda = \sum_{\bx \in [d]^t \setminus \distinct(d,t)} \proj{\bx} \,,
    \end{align*}
and decompose the projector into $O(t^2)$ sub-projectors according to which elements collide: since any tuple of non-distinct strings must have at least two equal entries, we have
\begin{align*}
    \bar \Lambda \leq \sum_{1\le i < j \le t} \Pi_{ij} \otimes \id_{[n]\setminus \{ij\}}\,, \quad \text{ where } \Pi_{ij} = \sum_{x \in [d]} \proj{x}_i \ot \proj{x}_j, 
    \end{align*}
where the subscript $i$ denotes the register $\reg{A}_i$. We shall omit the identity from the notation below.
    
Using cyclicity of trace, we have
    \begin{align*}
    \Tr[\bar \Lambda \E_{C} C^{\ot t} \rho C^{\ot t}] 
    &\leq \sum_{i < j} \Tr[\E_{C} C^{\ot t, \dagger} \Pi_{ij}  C^{\ot t} \rho] 
    = \sum_{i < j} \Tr[\E_{C} (C_i^\dagger \ot C_j^{\dagger})  \Pi_{ij} (C_i \ot C_j) \rho],  
\end{align*}
where for the second equality, we cancelled the $C$-unitaries on all systems except $i$ and $j$, with $C_i$ denoting application of $C$ on system $i$.

Using the standard fact that the Clifford group forms a 2-design, we can replace the average over Clifford unitaries with an average over the Haar measure in the above expression since we only use the second moment. Thus,
\begin{align*}
    &   \Tr[\bar \Lambda \E_{C} C^{\ot t} \rho C^{\ot t}]  \le \sum_{i < j} \Tr[\E_{U \sim \HaarMeasure} (U_i^\dagger \ot U_j^{\dagger}) \Pi_{ij}  (U_i \ot U_j) \rho] = \sum_{i < j} \Tr[\E_{U \sim \HaarMeasure} (U^\dagger \ot U^{\dagger}) \Pi_{ij} (U \ot U) \rho_{ij}],
\end{align*}
where for the last equality, we performed the partial trace over all systems except $i$ and $j$, with $\rho_{ij}$ denoting the reduced state on these systems. Since $\rho_{ij}$ is a quantum state, we can bound each of the $O(t^2)$ trace terms by the corresponding operator norms to obtain
    \begin{align}\label{eqn:2}
    \tr[\bar \Lambda \E_{C} C^{\ot t} \rho C^{\ot t}] 
    &\leq O(t^2d) \norm{\E_{U \sim \HaarMeasure} (U \ot U)^{\dagger} \left(\frac{\Pi}{\Tr[\Pi]}\right) (U \ot U)}_\infty \,,
    \end{align}
where $\Pi = \sum_{x\in [d]} \proj{x} \ot \proj{x}$ with $\Tr[\Pi]=d$. 

Since applying a Haar random unitary on any state gives a Haar random state $\ket{\psi} \in \C^d$, by linearity of expectation the expression inside the operator norm is 
\[ \E_{\ket{\psi} \sim \HaarMeasure}[\proj{\psi} \ot \proj{\psi}]. \]
It is well known that this is the maximally mixed state on the symmetric subspace $\text{Sym}^2(d)$ of $\C^d \ot \C^d$, which has dimension  $\frac{d(d+1)}{2}$ (see \cite{harrow2013church}, Proposition 6). Thus, the operator norm is $\frac{2}{d(d+1)}$, and inserting this into \Cref{eqn:2} yields the claimed result.

\subsection{Action of the $t$-wise Haar twirl (proof of \Cref{lem:haar-twirl})}
\label{sec:haar-twirl}

In order to derive the expression given by \Cref{lem:haar-twirl}, we first compute the result of applying the $t$-wise Haar twirl on Schur-Weyl basis states.

\begin{lemma} \label{lem:haar_moment_sw}
Let $\reg A \cong (\C^d)^{\ot t}$ and let $\ket{\alpha} = \ket{w_{\lambda, i}} \ot \ket{v_{\lambda, j}}$ and $\ket{\beta} = \ket{w_{\lambda', i'}} \ot \ket{v_{\lambda', j'}}$ be Schur-Weyl basis states on $\reg A$. 
Then 
\begin{align*}
\momhaart{\ketbra{\alpha}{\beta}} = \begin{cases} \frac{\id_{W_{\lambda}}}{\dim W_{\lambda}} \ot \ketbra{v_{\lambda, j}}{v_{\lambda, j'}}\, , & \text{ if $\lambda = \lambda'$ and $i = i'$} \\ 0 \, , & \text{otherwise} \end{cases}\,.
\end{align*}
\end{lemma}
\begin{proof}
By Schur-Weyl duality (\Cref{lem:schur-weyl}), we have that $U^{\ot t} = \sum_{\lambda_1} U^{(\lambda_1)} \otimes \id_{V_{\lambda_1}}$ where $U^{(\lambda_1)}$ only acts on $W_{\lambda_1}$. Thus, 
\begin{align*}
\momhaart{\ketbra{\alpha}{\beta}}  = \sum_{\lambda_1, \lambda_2} \left( \E_{U \sim \HaarMeasure} U^{(\lambda_1)} \ketbra{w_{\lambda, i}}{w_{\lambda', i'}} U^{(\lambda_2), \dagger} \right) \ot  \id_{V_{\lambda_1}}\ketbra{v_{\lambda, j}}{v_{\lambda', j'}}  \id_{V_{\lambda_2}} \,.
\end{align*}
We may assume that $\lambda = \lambda'$, since the above is zero otherwise. Then, we have that 
\begin{align*}
\momhaart{\ketbra{\alpha}{\beta}}  = \left( \E_{U \sim \HaarMeasure} U^{(\lambda)} \ketbra{w_{\lambda, i}}{w_{\lambda, i'}} U^{(\lambda), \dagger} \right) \ot \ketbra{v_{\lambda, j}}{v_{\lambda, j'}} \,.
\end{align*}
We claim that the expression inside the paranthesis above is zero unless $i=i'$, in which case it is the maximally mixed state on the subspace $W_{\lambda}$. This follows from a standard fact in representation theory called Schur's Lemma (see~\cite{fulton2013representation}, Lemma 1.7), which says that if $(\mu, H)$ is an irreducible representation of a group $G$ and $T: H \to H$ is a linear map such that $T \circ \mu = \mu \circ T$ (such a map is called an \emph{intertwiner}), then $T = \gamma \cdot \id_{H}$ for some scalar $\gamma \in \C$. Applying it to our setting, we see that the operator
$$
T = \E_{ U \sim \HaarMeasure} {U}^{(\lambda)} \ketbra{w_{\lambda, i}}{w_{\lambda, i'}} {U}^{(\lambda),\dag} $$ is an intertwiner for the irrep $(U^{(\lambda)}, W_{\lambda})$. This fact is a simple consequence of Haar invariance, since
\begin{align*}
T \tilde{U}^{(\lambda)} &= \left(\E_{U \sim \HaarMeasure} {U}^{(\lambda)} \ketbra{w_{\lambda, i}}{w_{\lambda, i'}} {U}^{(\lambda),\dag} \right) \tilde{U}^{(\lambda)} \\
&=\E_{ U \sim \HaarMeasure} {U}^{(\lambda)} \ketbra{w_{\lambda, i}}{w_{\lambda, i'}} \left(\tilde{U}^{(\lambda),\dag} {U}^{(\lambda)}\right)^\dag\\
&=\E_{U \sim \HaarMeasure} \left(\tilde{U}^{(\lambda)}{U}^{(\lambda)}\right) \ketbra{w_{\lambda, i}}{w_{\lambda, i'}} {U}^{(\lambda),\dag} = \tilde{U}^{(\lambda)} T.
\end{align*}

Therefore, by Schur's Lemma, 
\begin{align*}
\E_{U \sim \HaarMeasure} U^{(\lambda)} \ketbra{w_{\lambda, i}}{w_{\lambda, i'}} U^{(\lambda),\dag} = \gamma_{\lambda, i, i'} \id_{W_{\lambda}}
\end{align*}
for some scalars $\gamma_{\lambda, i, i'}$. Since the operator on the left is traceless if $i \neq i'$, we have that $\gamma_{\lambda, i, i'} = 0$ unless $i = i'$. For $i = i'$, the normalisation follows because the operator on the left is a mixed state, i.e.~it has unit trace. This completes the proof.
\end{proof}

We can now compute the result of applying a $t$-wise Haar twirl to a general state.
\begin{proof}[Proof of \Cref{lem:haar-twirl}]
Expanding in the Schur-Weyl basis, 
\begin{align*}
\ket{\phi}_{\reg{A \tilde E}} = \sum_{\lambda, i, j} \left( \ket{w_{\lambda, i}} \ket{v_{\lambda, j}} \right)_{\reg A} \ot \ket{e_{\lambda, i, j}}_{\reg {\tilde E}}
\end{align*}
for not necessarily normalised vectors $\ket{e_{\lambda, i, j}}_{\reg {\tilde E}}$. It follows from \Cref{lem:haar_moment_sw} and linearity that 
\begin{align*}
\momhaart{\proj{\phi}_{\reg{A \tilde E}}} &= \sum_{\lambda, i, j, j'} \left( 
\frac{\id_{W_{\lambda}}}{\dim W_{\lambda}} \ot \ketbra{v_{\lambda, j}}{v_{\lambda, j'}} \right)_{\reg A} \ot \ketbra{e_{\lambda, i, j}}{e_{\lambda, i, j'}}_{\reg {\tilde E}} \,
\\ & = \sum_{\lambda \vdash t} \rho_{\lambda} \otimes \Tr_{W_\lambda}[\ind_{P_\lambda} \proj{\phi}_{\reg{A\tilde E}}\ind_{P_\lambda}],
\end{align*}
where $\rho_\lambda$ is the maximally mixed state on the subspace $W_{\lambda}$. 
\end{proof}

\subsection{Permutation-phase twirl on distinct subspace (proof of \Cref{lem:pf-twirl})}
\label{sec:pf-twirl}

In order to derive an exact expression for the permutation-phase twirl on the distinct subspace with projector $\Lambda$, we start with the following lemma.

\begin{lemma} \label{lem:dis_string_to_perm}
Let $\bx, \by \in \distinct(d,t)$. 
Then
\begin{align*}
\mompft{ \ketbra{\bx}{\by}} = \begin{cases} \dfrac{\Lambda R_{\sigma}}{\Tr[\Lambda]}  & \text{ if } \by = \bx_{\sigma} \text{ for } \sigma \in S_t,\\ 0 & \text{ otherwise.} \end{cases} 
\end{align*}
\end{lemma}
Note that since $\bx, \by \in \distinct(d,t)$, there exists at most one permutation for which $\by = \bx_{\sigma}$. 
\begin{proof}
Applying the $t$-wise $F$-twirl first, we get that
\begin{align*}
\E_{F} F^{\ot t} \ketbra{\bx}{\by} F^{\ot t, \dagger}
&= \left( \E_f (-1)^{\sum_i f(x_i) + f(y_i)}\right) \ketbra{\bx}{\by} \,.
\end{align*}
Here, the expectation $\E_f$ is over a uniformly random function $f: [d] \to \bit$.
Because $\bx$ and $\by$ are both tuples of distinct strings, it is easy to see that
\begin{align*}
 \E_f \left[(-1)^{\sum_i f(x_i) + f(y_i)} \right] = 
 \begin{cases}
1, & \text{ if } \by = \bx_{\sigma}, \text{ for some } \sigma \in S_t, \text{ and }\\
0, & \text{ otherwise}.
 \end{cases}
\end{align*}
Next applying the $t$-wise $P$-twirl, 
\begin{align*}
\E_{P} P^{\ot t}  \left(\E_{F} F^{\ot t} \ketbra{\bx}{\bx} F^{\ot t, \dagger}\right) P^{\ot t, \dagger} 
&= \E_{P} P^{\ot t} \ketbra{\bx}{\bx}R_{\sigma} P^{\ot t, \dagger} = \E_{P} P^{\ot t}\proj{\bx} P^{\ot t, \dagger} R_\sigma \,,
\end{align*}
where we used that  $P^{\ot t}$ commutes with $R_\sigma$.
To conclude, we note that for any tuple of distinct strings $\bx$, 
\begin{equation*}
\ \E_{P} P^{\ot t} \proj{\bx} P^{\ot t, \dagger} = \frac{\Lambda}{\Tr[\Lambda]}. \qedhere
\end{equation*}
\end{proof}

We will also need the following technical result, which follows from the Schur orthogonality relations (\Cref{lem:orthogonality}).

\begin{lemma} \label{lem:K_simplified}
Let $\ket{\alpha} = \ket{w_{\lambda, i}} \ot \ket{v_{\lambda, j}}$ and $\ket{\beta} = \ket{w_{\lambda', i'}} \ot \ket{v_{\lambda', j'}}$ be Schur-Weyl basis states.
Then
\begin{align*}
\sum_{\sigma \in S_t} \bra{\beta} R^\dagger_{\sigma} \ket{\alpha} R_\sigma = \delta_{\lambda,\lambda'} \delta_{i,i'} \cdot \frac{t!}{\dim V_{\lambda}} \cdot (\id_{W_\lambda} \ot \ketbra{v_{\lambda,j}}{v_{\lambda, j'}})\,.
\end{align*}
\end{lemma}

\begin{proof}
We compute the matrix elements of $R^\dagger_\sigma$ in the Schur-Weyl basis. Recall that \Cref{lem:schur-weyl} implies that $R_{\sigma} = \sum_\lambda \id_{W_\lambda} \ot R^{(\lambda)}_\sigma$, where $R^{(\lambda)}_\sigma$ is the irreducible sub-representation of $R_\sigma$ on the Specht module $V_\lambda$. Thus, 
\begin{align}\label{eqn:perm-decomp}
(\bra{w_{\lambda', i'}}\bra{v_{\lambda', j'}}) R^\dagger_{\sigma} (\ket{w_{\lambda, i}} \ket{v_{\lambda, j}})
&= (\bra{w_{\lambda', i'}} \bra{v_{\lambda', j'}}) \left( \sum_\lambda \id_{W_\lambda} \ot R^{(\lambda), \dagger}_\sigma \right) (\ket{w_{\lambda, i}} \ket{v_{\lambda, j}}) \notag \\
&= \delta_{\lambda,\lambda'} \delta_{i,i'}  \cdot \bra{v_{\lambda, j'}} R^{(\lambda), \dagger}_\sigma \ket{v_{\lambda, j}} \notag\\
\ & =  \delta_{\lambda,\lambda'} \delta_{i,i'} \cdot \overline{\bra{v_{\lambda, j}} R^{(\lambda)}_\sigma \ket{v_{\lambda, j'}}} \,.
\end{align}
Therefore, for the rest of the proof, we consider $\ket{\alpha}$ and $\ket{\beta}$ with $\lambda = \lambda'$ and $i = i'$.

Again using the decomposition of $R_{\sigma}$ in terms of its irreducible sub-representations together with \Cref{eqn:perm-decomp}, we can write 
\begin{align*}
    \sum_{\sigma \in S_t} (\bra{w_{\lambda, i}}\bra{v_{\lambda, j'}}) R^\dagger_{\sigma} (\ket{w_{\lambda, i}} \ket{v_{\lambda, j}}) R_{\sigma} = \sum_{\lambda_1 \vdash t}  \id_{W_{\lambda_1}} \ot \left(\sum_{\sigma \in S_t} \overline{\bra{v_{\lambda, j}} R^{(\lambda)}_\sigma \ket{v_{\lambda, j'}}} R^{(\lambda_1)}_\sigma \right)
\end{align*}

Schur's orthogonality relations (\Cref{lem:orthogonality}) now imply that the operator in the paranthesis is zero unless $\lambda_1=\lambda$, in which case it equals 
\[ \sum_{\sigma \in S_t} \overline{\bra{v_{\lambda, j}} R^{(\lambda)}_\sigma \ket{v_{\lambda, j'}}} R^{(\lambda)}_\sigma  = \frac{t!}{\dim(V_{\lambda})} \ketbra{v_{\lambda,j}}{v_{\lambda,j'}}.\]
Plugging this in gives the desired result.
\end{proof}

We are now in a position to prove \Cref{lem:pf-twirl} which expresses the result of applying the permutation-phase twirl to any state that is only supported on distinct strings on $\reg A$ in terms of the Schur-Weyl subspaces. 

\begin{proof}[Proof of \Cref{lem:pf-twirl}] We first expand $\ket{\phi}$ in the standard basis on $\reg{A}$:
\begin{align}
\ket{\phi}_{\reg{A \tilde E}} = \sum_{\substack{\bx  \in \distinct(d,t)}} \ket{\bx}_{\reg A} \ket{\tilde e_{\bx}}_{\reg{ \tilde E}} \,, \label{eqn:std_basis_expansion}
\end{align}
where $\ket{\tilde e_{\bx}}_{\reg E}$'s are unnormalized and not necessarily orthogonal vectors and we used that $\ket{\phi}_{\reg{A \tilde E}}$ is supported over the distinct subspace in the register $A$. 

Applying the permutation-phase twirl to the register $\reg{A}$, we have by linearity, 
\begin{align*}
\mompft{\proj{\phi}_{\reg{AE}}} &= 
\sum_{\substack{\bx,\by \in \distinct(d,t)}} \left( \mompft{\ketbra{\bx}{\by}_{\reg  A}} \right) \ot \ketbra{\tilde e_{\bx}}{\tilde e_{\by}}_{\reg{\tilde E}} \,.
\end{align*}
\Cref{lem:dis_string_to_perm} implies that the term in the parantheses is non-zero only when $\by = \bx_{\sigma}$ for some $\sigma \in S_t$. Since we sum over all possible $\bx, \by \in \distinct(d,t)$ and there is at most one such $\sigma$ for each pair of tuples $\bx$ and $\by$, it follows that 
\begin{align}\label{eqn:pf1}
\mompft{\proj{\phi}_{\reg{AE}}} 
&= \sum_{\substack{\bx \in \distinct(d,t)\\ \sigma \in S_t}}  \frac{\left(\Lambda R_\sigma\right)_{\reg A}}{\Tr[\Lambda]}  \ot \ketbra{\tilde e_{\bx}}{\tilde e_{{\bx_{\sigma}}}}_{\reg{\tilde E}} \notag \\
&= \frac{\Lambda_{\reg A}}{\Tr[\Lambda]} \sum_{\sigma \in S_t}  (R_\sigma)_{\reg A} \ot \Bigg( 
\sum_{\bx \in \distinct(d,t)}  \ketbra{\tilde e_{\bx}}{\tilde e_{{\bx_{\sigma}}}}_{\reg{\tilde E}}  \Bigg) \notag \\
&= \frac{\Lambda_{\reg A}}{\Tr[\Lambda]} \sum_{\sigma \in S_t}  (R_\sigma)_{\reg A} \ot \Tr_{\reg A}\left[(R^\dagger_{\sigma} \ot \id_{\reg{\tilde E}}) \proj{\phi}_{\reg{A \tilde E}}\right]\,.
\end{align}
For the second equality we simply rearranged sums, and for the last equality we wrote the expression in parentheses in the second line more compactly as a partial trace, which follows directly from the expansion in \Cref{eqn:std_basis_expansion}.

We now rewrite the above partial trace in the Schur-Weyl basis on $\reg A$. Let 
\begin{align}
    \ket{\phi}_{\reg{A \tilde E}} = \sum_{\lambda,i,j} (\ket{w_{\lambda, i}} \ket{v_{\lambda, j}})_{\reg{A}} \ket{e_{\lambda,i,j}}_{\reg{ \tilde E}} \label{eqn:phi_schur_expansion}
\end{align}
be the state in the Schur-Weyl basis, where $\ket{e_{\lambda,i,j}}_{\reg{ \tilde E}}$ are unnormalized and not necessarily orthogonal vectors.
Then, 
\begin{align*}
\Tr_{\reg A}\left[(R^\dagger_{\sigma} \ot \id_{\reg{\tilde E}}) \proj{\phi}_{\reg{A \tilde E}}\right] =  \sum_{\substack{\lambda, i, j \\ \lambda', i', j'}} (\bra{w_{\lambda', i'}}\bra{v_{\lambda', j'}}) R^\dagger_{\sigma} (\ket{w_{\lambda, i}} \ket{v_{\lambda, j}}) \ketbra{e_{\lambda, i, j}}{e_{\lambda', i', j'}}_{\reg{\tilde E}}.
\end{align*}
Plugging the above into \Cref{eqn:pf1}, 
\begin{align*}
\mompft{\proj{\phi}_{\reg{A \tilde E}}}
&= \frac{\Lambda_{\reg A}}{\Tr[\Lambda]} \sum_{\substack{\lambda, i, j \\ \lambda', i', j'}} \left( \sum_{\sigma} (\bra{w_{\lambda', i'}}\bra{v_{\lambda', j'}}) R^\dagger_{\sigma} (\ket{w_{\lambda, i}} \ket{v_{\lambda, j}}) R_\sigma \right)_{\reg A} \ot \ketbra{e_{\lambda, i, j}}{e_{\lambda', i', j'}}_{\reg{\tilde E}} \,.
\end{align*}
Applying \Cref{lem:K_simplified} to the term in parentheses, we can simplify this to 
\begin{align}\label{eqn:pf2}
\mompft{\proj{\phi}_{\reg{A \tilde E}}}
&= \frac{\Lambda_{\reg A}}{\Tr[\Lambda]} \sum_{\substack{\lambda, i, j, j'}} \frac{t!}{\dim (V_{\lambda})} \left(\id_{W_\lambda} \ot \ketbra{v_{\lambda,j}}{v_{\lambda, j'}}\right)_{\reg A} \ot \ketbra{e_{\lambda, i, j}}{e_{\lambda, i, j'}}_{\reg {\tilde E}}\,.
\end{align}
By \Cref{lem:perm_inv_sw}, we can write $\Lambda = \bigoplus_{\lambda \vdash t} \Lambda^{(\lambda)}_{W_{\lambda}} \ot \id_{V_{\lambda}}$, where $\Lambda^{(\lambda)}_{W_\lambda}$ are projectors supported on $W_{\lambda}$. Since $\id_{W_\lambda}$ is simply the identity on subspace $W_\lambda$, this implies 
\begin{align*}
\Lambda (\id_{W_\lambda} \ot \ketbra{v_{\lambda,j}}{v_{\lambda, j'}}) = \Lambda^{(\lambda)}_{W_\lambda} \ot \ketbra{v_{\lambda,j}}{v_{\lambda, j'}} \,.
\end{align*}
Plugging this into \Cref{eqn:pf2} and rewriting it as a partial trace,
\begin{align*}
\mompft{\proj{\phi}_{\reg{A \tilde E}}}
    & = \sum_{\lambda \vdash t}  \dfrac{t!}{\dim(V_{\lambda}) \Tr[\Lambda]} \Lambda^{(\lambda)}_{W_\lambda} \otimes \Tr_{W_\lambda}[\ind_{P_\lambda} \proj{\phi}_{\reg{A \tilde E}}\ind_{P_\lambda}] \,.
\end{align*}
Since $\Tr[\Lambda^{(\lambda)}_{W_\lambda}] = (\dim(V_{\lambda}) \Tr[\Lambda])/t!$ as we showed in \Cref{clm:distinct-trace}, the first tensor factor is indeed the maximally mixed state on the support of the projector $\Lambda^{(\lambda)}_{W_\lambda}$. This completes the proof.
\end{proof}

\section{Efficient unitary $t$-designs}
\label{sec:linear-t-design}

As explained in \Cref{sec:intro_t_designs}, to turn the $PFC$ ensemble from \Cref{thm:pfc-ensemble} into an ensemble of efficiently implementable unitaries, we can replace the random functions and permutations in the definition of the $F$ and $P$ random unitaries by their $O(t)$-wise independent counterparts.
This will yield a very efficient $t$-design construction.
We first recall the definitions of $t$-wise independent functions and permutations.

\begin{definition}[$t$-wise independent functions]\label{def:t-wise-ind-func}
A distribution $\mathcal{D}$ over functions $\mathcal{F} = \{f: [N] \rightarrow [M]\}$ is called $t$-wise independent if, for all distinct $x_1,\dots,x_t \in [N]$ and for all $y_1,\dots,y_t \in [M]$ it holds that
$$
\Pr_{f \sim \mathcal{D}} \big[ f(x_1)=y_1 \, \land \, \dots \, \land \, f(x_t)=y_t 
\big] = M^{-t}.
$$
Moreover, we say that an ensemble of functions $\mathcal{F}$ is $t$-wise independent if the uniform distribution over the set $\mathcal{F}$ is $t$-wise independent.
\end{definition}

\begin{definition}[$t$-wise independent permutations]\label{def:t-wise-ind-perm}
A distribution $\mathcal{D}$ over permutations $\mathcal{P} = \{\pi: [N] \rightarrow [N]\}$ is called $t$-wise independent if, for all distinct $x_1,\dots,x_t \in [N]$ and all distinct $y_1,\dots,y_t \in [N]$, it holds that
$$
\Pr_{\pi \sim \mathcal{D}} \big[ \pi(x_1)=y_1 \, \land \, \dots \, \land \, \pi(x_t)=y_t 
\big] = \left(N \cdot (N-1) \cdots (N-t+1) \right)^{-1}.
$$
Moreover, we say that $\mathcal{D}$ is $\delta$-approximate $t$-wise independent\footnote{We remark that our definition of $\delta$-approximate $t$-wise independence is the same as in \cite{v009a015}.} if instead
$$
\Big| \Pr_{\pi \sim \mathcal{D}} \big[ \pi(x_1)=y_1 \, \land \, \dots \, \land \, \pi(x_t)=y_t 
\big] - \left(N \cdot (N-1) \cdots (N-t+1) \right)^{-1}
\big] \Big| \leq \delta.
$$
\end{definition}

We will use the following theorem due to Alon and Lovett~\cite{v009a015}, which says that an \emph{approximate} $t$-wise independent distribution over permutations is always statistically close to an exact $t$-wise independent distribution over permutations---provided that the error is sufficiently small.

\begin{theorem}[Theorem 1.1, \cite{v009a015}]\label{thm:alon}
Let $\mathcal{D}$ be a distribution over permutations $\mathcal{P} = \{\pi: [N] \rightarrow [N]\}$ which is $\delta$-approximate $t$-wise independent. Then, there exists a distribution $\mathcal{D}'$ over $\mathcal{P}$ which is exactly $t$-wise independent such that $\| \mathcal{D} - \mathcal{D}' \|_1 \leq O\left(\delta \cdot N^{4t} \right)$.

\end{theorem}

Our efficient unitary $t$-design ensemble is a product of a random Clifford unitary, a $2t$-wise independent binary
phase operator, and an \emph{approximate} $t$-wise independent permutation operator. To prove this, we make use of $O(t)$-wise independence in order to switch to the ``fully random'' $PFC$ ensemble from \Cref{sec:pru_security}. One way to argue that this switch is justified is to invoke a theorem due to Zhandry~\cite[Theorem 3.1]{Zhandry2012} which says that the behavior of any quantum algorithm making at most $t$ queries to a
$2t$-wise independent function is identical to its behavior when querying a uniformly random function. While Zhandry's result allows us to immediately switch from a $2t$-wise independent binary-phase operator to a fully random binary-phase operator, the switch from \emph{approximate} $t$-wise independent permutations is  more subtle, especially so because the permutation operator is applied in-place rather than as a regular reversible oracle. 
It seems very likely that Zhandry's general result can be extended to the case of approximately independent permutations, but since we do not need the full power of Zhandry's result, we opt for a more direct and self-contained approach: we show that we can directly work with the $O(t)$-wise independence of the underlying function and permutation in order to carry out a similar analysis as in the ``fully random'' case in \Cref{sec:pru_security}. 
Using this direct approach, we show the following theorem.

\begin{theorem}[Efficient unitary t-design] \label{thm:t-designs}
Let $\mathcal{F} = \{f: [d] \rightarrow \bit\}$ be a $2t$-wise independent function family and $\mathcal{D}$ be a distribution over permutations $\mathcal{P} = \{\pi: [d] \rightarrow [d]\}$ which is $\delta$-approximate $t$-wise independent for 
 $\delta = O(t/d^{4t+\frac{1}{2}})$. Then, the unitary ensemble $\nu$ over $\mathrm{U}(d)$ which samples $U \sim \nu$ as a product
 $$
 U = P_\pi \, F_f \, C \,, \quad\quad \text{ where } \,\quad \pi \sim \mathcal{D}, \,  f \sim \mathcal{F}, \, C \sim \mathrm{C}(d)\, ,
 $$
 is a diamond $\eps$-approximate $t$-design for $\eps = O(t/\sqrt{d})$ (as per \Cref{def:additive_error_design}).
\end{theorem}
\begin{proof}
First, we argue that we can replace the $\delta$-approximate $t$-wise independent distribution $\mathcal{D}$ over $\mathcal{P}$ in our ensemble $\nu$ with a distribution which is exactly $t$-wise independent; specifically, we show that this incurs a small error in terms of diamond distance. To see this, we first use \Cref{thm:alon} to argue that there exists a distribution $\mathcal{D}'$ over $\mathcal{P}$ which is exactly $t$-wise independent such that $\| \mathcal{D} - \mathcal{D}' \|_1 \leq O\left(\delta \cdot d^{4t} \right)$. Let $\nu$ be our ensemble of unitaries and let $\nu'$ be the same ensemble, except that we replace $\mathcal{D}$ with $\mathcal{D}'$. Then,
\begin{align*}
\left\| \mathcal{M}_{\nu}^{(t)} - \mathcal{M}_{\nu'}^{(t)} \right\|_\Diamond &= \max_{\ket{\psi}_{\reg{AE}}}\left\| \E_{\substack{\pi\sim \mathcal{D}\\
f \sim \mathcal{F}\\
C \sim \mathrm{C}(d)}} (P_\pi F_f C)^{\ot t}_{\reg A} \proj{\psi}_{\reg{AE}} (P_\pi F_f C)^{\ot t, \dagger}_{\reg A} - \E_{\substack{\pi\sim \mathcal{D}\\
f \sim \mathcal{F}\\
C \sim \mathrm{C}(d)}}  (P_\pi F_f C)^{\ot t}_{\reg A} \proj{\psi}_{\reg{AE}} (P_\pi F_f C)^{\ot t, \dagger}_{\reg A}  \right\|_1\\
&= \max_{\ket{\psi}_{\reg{AE}}}\left\| \E_{\pi\sim \mathcal{D}} \xi^{\pi, \psi}_{\reg{AE}} - \E_{\pi\sim \mathcal{D}'} \xi^{\pi, \psi}_{\reg{AE}}  \right\|_1\, ,
\end{align*}
where we define the ensemble of density matrices $\{\xi^{\pi, \psi}_{\reg{AE}}\}$ as
$$
\xi^{\pi, \psi}_{\reg{AE}} :=  \E_{\substack{
f \sim \mathcal{F}\\
C \sim \mathrm{C}(d)}}  (P_\pi F_f C)^{\ot t}_{\reg A} \proj{\psi}_{\reg{AE}} (P_\pi F_f C)^{\ot t, \dagger}_{\reg A}.
$$
Then, by the strong convexity of the trace distance (see~\cite[Theorem 9.3]{Nielsen_Chuang_2010}), we have that
\begin{align*}
\left\| \mathcal{M}_{\nu}^{(t)} - \mathcal{M}_{\nu'}^{(t)} \right\|_\Diamond &=\max_{\ket{\psi}_{\reg{AE}}}\left\| \E_{\pi\sim \mathcal{D}} \xi^{\pi, \psi}_{\reg{AE}} - \E_{\pi\sim \mathcal{D}'} \xi^{\pi, \psi}_{\reg{AE}}  \right\|_1 \leq \big\| \mathcal{D} - \mathcal{D}' \big\|_1 \leq O\left(\delta \cdot d^{4t} \right) = O\left(\frac{t}{\sqrt{d}}\right) \, ,
\end{align*}
where we used \Cref{thm:alon} and that $\mathcal{D}$ is $\delta$-approximate with parameter $\delta =O(t/d^{4t+\frac{1}{2}})$.
Because of the triangle inequality, it suffices to show that the exact ensemble $\nu'$ satisfies
$$\left\| \mathcal{M}_{\nu'}^{(t)} - \mathcal{M}_{\text{Haar}}^{(t)} \right\|_\Diamond \leq O\left(\frac{t}{\sqrt{d}}\right).
$$
Next, we argue that we can carry out essentially the same proof as in the ``fully random'' ensemble.
For this, notice that the only place in the proof of \Cref{thm:pfc-ensemble} where we use properties of the $F$ and $P$ unitaries is \Cref{lem:dis_string_to_perm}.
Therefore, it suffices to show that \Cref{lem:dis_string_to_perm} still holds for the exact ensemble $\nu'$ which uses $O(t)$-wise independent functions and permutations rather than their ``fully random'' counterparts.
In other words, we can complete the proof of \Cref{thm:t-designs} by showing \Cref{lem:twise-lemma-exact}, which we prove below.
\end{proof}

\begin{lemma}[Restatement of \Cref{lem:dis_string_to_perm} for $t$-wise independence.] \label{lem:twise-lemma-exact}
Consider the PF ensemble, where $F$ is chosen from an ensemble $\mathcal{F} = \{f: [d] \rightarrow \bit\}$ of $2t$-wise independent functions and where $P$ is sampled according to an exact $t$-wise independent distribution over permutations $\mathcal{P} = \{\pi: [d] \rightarrow [d]\}$. Let $\bx, \by \in \distinct(d,t)$ be arbitrary. 
Then,
\begin{align*}
\mathcal{M}_{PF}^{(t)}(\ketbra{\bx}{\by}) = \begin{cases} \dfrac{\Lambda R_{\sigma}}{\Tr[\Lambda]}  & \text{ if } \by = \bx_{\sigma} \text{ for } \sigma \in S_t,\\ 0 & \text{ otherwise.} \end{cases} 
\end{align*}
\end{lemma}
\begin{proof}
Applying the $2t$-wise $F$-twirl first, we get that
\begin{align*}
\E_{F} F^{\ot t} \ketbra{\bx}{\by} F^{\ot t, \dagger}
&= \left( \E_{f \sim \mathcal{F}} (-1)^{\sum_i f(x_i) + f(y_i)}\right) \ketbra{\bx}{\by} \,.
\end{align*}
Here, the expectation is over the $2t$-wise independent function family $\mathcal{F}$.
Because $\bx, \by \in \distinct(d,t)$, it immediately follows from the $2t$-wise independence of $\mathcal{F}$ that
\begin{align*}
\E_{f \sim \mathcal{F}} \left[(-1)^{\sum_i f(x_i) + f(y_i)} \right]=
 \E_f \left[(-1)^{\sum_i f(x_i) + f(y_i)} \right] = 
 \begin{cases}
1, & \text{ if } \by = \bx_{\sigma}, \text{ for some } \sigma \in S_t, \text{ and }\\
0, & \text{ otherwise} \, ,
 \end{cases}
\end{align*}
where the second expectation is over ``fully random'' functions $f$.
Next applying the $t$-wise $P$-twirl, 
\begin{align*}
\E_{P} P^{\ot t}  \left(\E_{F} F^{\ot t} \ketbra{\bx}{\bx} F^{\ot t, \dagger}\right) P^{\ot t, \dagger} 
&= \E_{P} P^{\ot t} \ketbra{\bx}{\bx}R_{\sigma} P^{\ot t, \dagger} = \E_{P} P^{\ot t}\proj{\bx} P^{\ot t, \dagger} R_\sigma \,,
\end{align*}
where we used that  $P^{\ot t}$ commutes with $R_\sigma$.
Finally, using that $\mathcal{D}$ is a distribution over $t$-wise independent permutations over $\mathcal{P}$, we get that for any tuple of distinct strings $\bx$, 
\begin{align*}
\ \E_{P} P^{\ot t} \proj{\bx} P^{\ot t, \dagger} &= \E_{\pi \sim \mathcal{D}} \proj{\pi(x_1), \dots, \pi(x_t)}\\
&= \,\,\, \E_{\pi} \proj{\pi(x_1), \dots, \pi(x_t)} =
\frac{\Lambda}{\Tr[\Lambda]} \, , 
\end{align*}
where the last expectation is over ``fully random'' permutations over $[d]$. This proves the claim.
\end{proof}

\paragraph{Linear-depth $t$-design construction.} We can instantiate the $PFC$ ensemble in \Cref{thm:t-designs} with a concrete ensemble of $O(t)$-wise independent functions and permutations.
Very efficient $O(t)$-wise independent functions and permutations have been constructed, allowing us to get very efficient $t$-designs. 

Concretely, let us say that an ensemble $\mathcal{F} = \{f: [N] \rightarrow \bit\}$ of functions is $(s,r)$-explicit if there is a circuit of size at most $s$ and depth at most $r$ that evaluates any function $f \in \mathcal{F}$ on a given input. Analogously, we say that a distribution $\mathcal{D}$ over permutations $\mathcal{P} = \{\pi: [N] \rightarrow [N]\}$ is $(s,r)$-explicit if there is a circuit of size at most $s$ and depth at most $r$ that evaluates $\pi \sim \mathcal{D}$ on a given input. Then, we show the following.

\begin{lemma}[Explicit $t$-design] \label{cor:t-designs-explicit} Let $d=2^n$. Let $\mathcal{F} = \{f: [d] \rightarrow \bit\}$ be an $(s,r)$-explicit family of $2t$-wise independent functions and let $\mathcal{D}$ be an $(s,r)$-explicit, $\delta$-approximate $t$-wise independent distribution over permutations $\mathcal{P} = \{\pi: [d] \rightarrow [d]\}$ for 
 $\delta = O(t/d^{4t+\frac{1}{2}})$. Then, the resulting $n$-qubit ensemble $\nu$ from \Cref{thm:t-designs} is a diamond $\eps$-approximate $t$-design for $\eps = O(t/\sqrt{d})$, and each $U \sim \nu$ can be implemented in size $O(n^2 + s)$ and depth $O(n + r)$. 
\end{lemma}
\begin{proof}
Note that the depth which is required to implement a unitary $U$ from the $PFC$ ensemble $\nu$ is captured by the depth of implementing a random Clifford, a binary-phase operator (with respect to $\mathcal{F}$), and a permutation operator (with respect to $\mathcal{P}$).
Implementing any $n$-qubit Clifford unitary requires size $O(n^2)$ and depth of $O(n)$~\cite{Bravyi_2021}. Because Toffoli gates are universal for classical computation~\cite{Nielsen_Chuang_2010}, the unitaries $U_f: \ket{x}\ket{y} \rightarrow \ket{x}\ket{y \oplus f(x)}$ and $R_\pi: \ket{x} \rightarrow \ket{\pi(x)}$ can also be implemented in size $O(s)$ and depth $O(r)$ ---potentially using using ancilla qubits initialized to $\ket{0}$. Therefore, both the binary phase operator $F$ (consisting of a layer of Hadamard gates followed by Toffoli gates) as well as the permutation operator $P$ (consisting solely of Toffoli gates) can be implemented in size $O(s)$ and depth $O(r)$.
\end{proof}

We remark that since one can use any $2$-design instead of a Clifford in our construction, the above parameters can be optimized even further. For instance, it is known that one can sample a $2$-design with circuits of quasilinear size and logarithmic depth (see~\cite{Cleveetal16}). We do not try to optimize the parameters here as the dominant parameters come from the construction of $t$-wise independent functions and permutations. 

For $O(t)$-wise independent functions, we can use the following result.

\begin{fact}[\cite{Joffe,ALON1986567}]\label{fact:func}
For any $1 \leq t \leq N$, there exists an ensemble $\mathcal{F} = \{f: [N] \rightarrow \bit\}$ of $t$-wise independent functions which can be evaluated in size and depth $O(t \log N)$.
\end{fact}

Note that if the domain is the space of $n$-qubit computational basis states, then $N = 2^n$ and the above gives a $(s,r)$-explicit family of $2t$-wise independent functions where $s = O(tn)$ and $r = O(tn)$. 

For approximate $O(t)$-wise independent permutations, existing constructions are more complicated and less explicit than for $O(t)$-wise independent functions. 
A very elegant construction of $O(t)$-wise independent elements of the alternating group (i.e.~the even permutations) on $p^3-1$ elements (for a sufficiently large prime $p$) can be found in~\cite[Theorem 1.5]{caprace2023tame}, which can be turned into $O(t)$-independent permutations on any (sufficiently large) number of elements $N$ using techniques from~\cite{Kassabov05}.
This construction only involves basic finite-field arithmetic, yielding $O(t)$-wise independent permutations with circuit size $O(t~ \mathrm{polylog}(N))$. 
Alternatively, one can use a matrix-based construction from~\cite{Kassabov05}.
For this,~\cite{v009a015} claimed (informally without proof) that the permutations can be implemented in time $O(t \log N)$, but upon closer inspection it is not obvious how to achieve this without additional $\mathrm{polylog}(N)$-factors. 

\begin{fact}[\cite{Kassabov05,caprace2023tame,random_walks}]\label{fact:perm}
For any sufficiently large $N \in \N$ and for any constant $C>0$, there exists a $\delta$-approximate $t$-wise independent distribution $\mathcal{D}$ over permutations $\mathcal{P} = \{\pi: [N] \rightarrow [N]\}$ with $\delta = O(N^{-Ct})$ such that each $\pi \sim \mathcal{D}$ can be evaluated in size $O(t~ \mathrm{polylog}(N))$.
\end{fact}

An upcoming work~\cite{random_walks} analyses the construction from~\cite{Kassabov05}  in detail and presents explicit circuits, showing that the circuits can be made to have size $O(t \log^2(N))$ and depth $O(t \log N \log \log N)$ (with ancillas).
In particular, for the domain of $n$-qubit computational basis states, using the result from~\cite{random_walks} gives an $(s,r)$-explicit family of approximate $O(t)$-wise independent permutations with $s = O(t ~\poly(n))$ and depth $r = \tilde{O}(t n)$, where $\tilde{O}$ hides logarithmic factors in $n$. 

Using these aforementioned constructions of $O(t)$-wise independent functions and permutations as a black box, we obtain a $t$-design on $n$-qubits with the same size and depth as well. 
\begin{corollary}[Linear-depth $t$-design] \label{cor:t-designs-linear} Let $d=2^n$. Let $\mathcal{F} = \{f: [d] \rightarrow \bit\}$ be the $2t$-wise independent function family from \Cref{fact:func} and let $\mathcal{D}$ be the $\delta$-approximate $t$-wise independent distribution over permutations $\mathcal{P} = \{\pi: [d] \rightarrow [d]\}$ in \Cref{fact:perm} for 
 $\delta = O(t/d^{4t+\frac{1}{2}})$. Then, the resulting $n$-qubit ensemble $\nu$ from \Cref{thm:t-designs} is a diamond $\eps$-approximate $t$-design for $\eps = O(t/\sqrt{d})$, and each $U \sim \nu$ can be implemented in depth $O(t \, \poly(n))$. 
\end{corollary}
As mentioned before, the $n$-dependence of the depth can be improved to quasilinear using~\cite{random_walks}, yielding a depth of $\tilde O(tn)$. 

\subsection{Amplification of approximation error} \label{sec:t_design_amplification}

So far, we have constructed $n$-qubit $t$-designs with diamond error $\eps = O(t/2^{n/2})$.
By repeating our construction $m$ times independently in sequence, we can make the error decay like $\eps^m$.
Choosing $m = \Theta(t)$, this pushes down the error far enough that we can apply \Cref{lem:diamond-to-rel} to convert our diamond-error $t$-designs into relative-error ones.
Of course these repetitions come at the cost of increasing the size and the depth of the circuits: choosing $m = \Theta(t)$ introduces an additional factor $t$ in size and depth, which is why we obtain relative-error $t$-designs with quadratic scaling in $t$.
These amplification techniques are standard in the $t$-design literature (see e.g.~\cite{mele2023introduction}), but we spell out some of the details for completeness.

We begin with a simple auxiliary lemma about the concatenation of Haar moment operators with other moment operators, which follows immediately from the invariance of the Haar measure. 
\begin{lemma} \label{lem:haar_collapse}
Let $X_1, \dots, X_m$ be a collection of random matrices.
Suppose that at least one of the $X_i$ is independent and Haar random.
Then 
\begin{align*}
\cM^{(t)}_{X_m} \circ \cM^{(t)}_{X_m} \circ \cdots \circ \cM^{(t)}_{X_2} \circ \cM^{(t)}_{X_1}(\cdot) = \momhaart{\cdot} \,. 
\end{align*}
\end{lemma}

With this, we can show that concatenating independent samples from a $t$-design results in exponential decay of the error.
\begin{lemma} \label{lem:diamond_amp}
Suppose that $X \sim \cX$ is a diamond $\eps$-approximate $t$-design.
Let $X_1, \dots, X_m$ be independent random unitaries sampled from $\cX$.
Then $X_1 \cdots X_m$ is a diamond $\eps^m$-approximate $t$-design.
\end{lemma}
\begin{proof}
The moment operator for the random unitary $X_1 \cdots X_m$ can be written as $(\cM^{(t)}_X)^{\circ m} \deq \underbrace{\cM^{(t)}_X \circ \dots \circ \cM^{(t)}_X}_{\text{$m$ times}}$.
Observe that 
\begin{align*}
(\cM^{(t)}_X - \cM^{(t)}_{\Haar})^{\circ m} = (\cM^{(t)}_X)^{\circ m} + \sum_{i = 1}^m \binom{m}{i} (-1)^i \cM^{(t)}_{\Haar} = (\cM^{(t)}_X)^{\circ m} - \cM^{(t)}_{\Haar} \,.
\end{align*}
The first equality uses the fact that when expanding out the product, all terms except $(\cM^{(t)}_X)^{\circ m}$ have at least one copy of $\cM^{(t)}_{\Haar}$ in them, which allows us to apply \Cref{lem:haar_collapse}.
The lemma now follows directly from the submultiplicativity of the diamond norm (\Cref{eqn:diamond_submult}).
\end{proof}

Applying \Cref{lem:diamond_amp} and \Cref{lem:diamond-to-rel} to the $t$-designs from \Cref{cor:t-designs-explicit}, we get the following.

\begin{corollary} \label{lem:error_amplified}
Let $\nu$ be the ensemble of $n$-qubit unitaries from \Cref{cor:t-designs-explicit} (instantiated with $(s,r)$-explicit families of functions and permutations).
Let $U \sim \nu^{\circ m}$ be the ensemble of $m$-fold products of unitaries from $\nu$, i.e.~$U = U_1 \cdots U_m$ for $U_i \sim \nu$.
Then, each $U \sim \nu^{\circ m}$ can be implemented in size $O(m \cdot(n^2 + s))$ and depth $O(m \cdot(n+r))$ and the ensemble $\nu^{\circ m}$ is \begin{enumerate}
\item a diamond $\eps$-approximate $t$-design with $\eps = O(t^m \cdot 2^{-nm/2})$, and 
\item a relative-error $\eps$-approximate $t$-design with $\eps = O(t^m \cdot 2^{2nt - nm/2})$.
\end{enumerate}
In particular, this means that for all $\eps > 0$ and $t \leq 2^{-n/4}$, we have \begin{enumerate}
\item diamond $\eps$-approximate $t$-designs in size $O(t (n^2 + s) + t \log1/\eps)$ and depth $O(t (n + r) + t \log1/\eps)$,
\item relative-error $\eps$-approximate $t$-designs in size $O(t^2 (n^2 + s) + t \log1/\eps)$ and depth $O(t^2 (n + r) + t^2 \log1/\eps)$.
\end{enumerate}
\end{corollary}
\begin{proof}
The first two statements follow immediately from \Cref{cor:t-designs-linear} by means of  \Cref{lem:diamond_amp} and \Cref{lem:diamond-to-rel}.
The second two statements follow from the first two simply by choosing $m = \max\{1, O(\frac{\log 1/\eps}{n})\}$ large enough that the desired $\eps$ is achieved.
\end{proof}

Again, using the upcoming work~\cite{random_walks} the depth can be made quasilinear in $n$.

\section{Pseudorandom unitaries with non-adaptive security} \label{sec:pru_formal}

We first give a formal definition of PRUs, as proposed by Ji, Liu, and Song~\cite{ji2018pseudorandom}. Then, we prove that if we simply replace the random function in the $F$-operator and the random permutation in the $P$-operator by their pseudorandom counterparts, the resulting ensemble (described in \Cref{eq:intro-pru}) is a non-adaptive pseudorandom unitary.

\begin{definition}[Pseudorandom unitary]\label{def:PRU} Let $n \in \N$ be the security parameter. An infinite sequence $\mathcal{U} = \{\mathcal{U}_n\}_{n \in \N}$ of $n$-qubit unitary ensembles $\mathcal{U}_n = \{U_k\}_{k \in \mathcal{K}}$ is a pseudorandom unitary if it satisfies the following conditions.
\begin{itemize}
    \item (Efficient computation) There exists a polynomial-time quantum algorithm $\mathcal{Q}$ such that for all keys $k \in \mathcal{K}$, where $\mathcal{K}$ denotes the key space, and any $\ket{\psi} \in (\C^2)^{\ot n}$, it holds that
    $$
    \mathcal{Q}(k,\ket{\psi}) = U_k \ket{\psi}\,.
    $$

    \item (Pseudorandomness) The unitary $U_k$, for a random key $k \sim \algo K$, is computationally indistinguishable from a Haar random unitary $U \sim \HaarMeasure(2^n)$. In other words,  for any QPT algorithm $\algo A$, it holds that $$
    \vline\, \underset{k \sim \algo K}{\Pr}[\algo A^{U_k}(1^\lambda)=1] - \underset{U \sim \Haar}{\Pr}[\algo A^{U}(1^\lambda) =1]  \,\vline \,\leq \, \negl(n)\,.
    $$ 

We call $\mathcal{U} = \{\mathcal{U}_n\}_{n \in \N}$ a \emph{non-adaptive} pseudorandom unitary if $\algo A$ is only allowed to make parallel queries to the unitary $U_k$ (or $U$ in the Haar random case).
\end{itemize}
Note that, whenever we write $\mathcal{U}_n = \{U_k\}_{k \in \mathcal{K}}$, it is implicit that the key space $\mathcal{K}$ depends on the security parameter $n \in \N$, and that the length of each key $k \in \mathcal{K}$ is polynomial in $n$.
\end{definition}

The main result of this section is that the construction in \Cref{eq:intro-pru} is indeed a non-adaptive PRU.

\begin{theorem} \label{thm:pru_security}
Let $n \in \N$ be the security parameter. Then, the ensemble
$\mathcal{U}_n = \{U_k\}_{k \in \mathcal{K}}$ of $n$-qubit unitary operators defined in \Cref{eq:intro-pru} is a non-adaptive pseudorandom unitary when instantiated with ensembles of $n$-bit (quantum-secure) PRFs and PRPs.
\end{theorem}

\begin{proof}
From the construction, it is clear that a random unitary $U_k$ from the above family can be sampled efficiently (see e.g.~\cite{berg2021simple} for simple way to sample a uniform Clifford unitary). To argue security against any non-adaptive algorithm making $t= \poly(n)$ queries, it suffices to show that for any initial state $\ket{\psi}_{\reg{A} \reg{E}}$, where register  $\reg A \cong ((\C^2)^{\ot n})^{\ot t}$ is on $nt$ qubits, and $\reg{E}$ is an arbitrary workspace register, the density matrices
\[ \rho:= \E_{k \in \cK} (U_k)^{\ot t}_{\reg A} \proj{\psi}_{\reg{AE}} (U_k)^{\ot t, \dagger}_{\reg A} \text{\hphantom{t} and } \rho^{\rm hr}: = \E_{U \sim \Haar} U^{\ot t}_{\reg A} \proj{\psi}_{\reg{AE}} U^{\ot t, \dagger}_{\reg A}, \]
are computationally indistinguishable with at most negligible advantage, since that is the general form of a non-adaptive distinguisher. From the post-quantum security of the PRF and PRP families assumed in \Cref{eq:intro-pru}, it follows immediately that if we replace the pseudorandom permutation and function with their fully random counterparts to obtain the ``fully random'' state $\rho^{\rm fr}$, then $\rho^{\rm fr}$ is computationally indistinguishable from $\rho$ up to negligible advantage in $n$. Furthermore, since the $PFC$ ensemble is a diamond-distance $t$-design, using the error bounds from \Cref{thm:pfc-ensemble}, it follows that $\|\rho^{\rm fr} - \rho^{\rm hr}\|_1 \le O(t/\sqrt{2^n})$. This is also negligible in $n$ since $t=\poly(n)$ and the computational indistinguishability of $\rho$ and $\rho^{\rm hr}$ follows. 
\end{proof}

\section{Pseudorandom isometries with adaptive security} \label{sec:pri_adaptive_security}

One drawback of our PRU construction in \cref{sec:pru_formal} is that we are only able to prove non-adaptive security.
This is a consequence of the fact that the analysis of the $PFC$ ensemble achieves diamond-error, not relative error.
In contrast to the $t$-design amplification in \cref{sec:t_design_amplification}, we cannot simply amplify our PRU construction to achieve relative error.
This is because the number of iterations in the amplification in \cref{sec:t_design_amplification} depends on the number of queries $t$, but for our PRUs we do not have an a priori bound on $t$.

However, it turns that if we relax the notion of PRUs to PRIs, a simple modification of our construction is able to achieve adaptive security.
We have already given a high-level of this construction in \cref{sec:intro_adaptive_pri}, so we proceed with the formal statements here.

The formal definition of PRIs is entirely analogous to \cref{def:PRU}, but we spell it out again for the sake of completeness.

\begin{definition}[Pseudorandom isometry]\label{def:PRI} Let $n \in \N$ be the security parameter and choose an integer function $s(n) \in [0,n)$. An infinite sequence $\mathcal{V} = \{\mathcal{V}_n\}_{n \in \N}$ of ensembles $\mathcal{V}_n = \{V_k: \C^{2^{n-s(n)}} \to \C^{2^n}\}_{k \in \mathcal{K}}$ of isometries from $n - s(n)$ qubits to $n$ qubits is a pseudorandom isometry if it satisfies the following conditions.
\begin{itemize}
    \item (Efficient computation) There exists a polynomial-time quantum algorithm $\mathcal{Q}$ such that for all keys $k \in \mathcal{K}$, where $\mathcal{K}$ denotes the key space, and any $\ket{\psi} \in (\C^2)^{\ot n - s(n)}$, it holds that
    $$
    \mathcal{Q}(k,\ket{\psi}) = V_k \ket{\psi}\,.
    $$

    \item (Pseudorandomness) The isometry $V_k$, for a random key $k \sim \algo K$, is computationally indistinguishable from a Haar random isometry $V \sim \HaarMeasure(2^{n-s(n)}, 2^n)$.\footnote{By a Haar random isometry $V \sim \HaarMeasure(2^{n-s(n)}, 2^n)$ we mean an isometry sampled as follows. First sample a Haar random unitary $U \sim \HaarMeasure(2^n)$ on $n$ qubits, then set $V \ket{\psi} \deq U(\ket{\psi}\ket{0}^{\ot s(n)})$.
    In other words, to apply a Haar random isometry to an $(n-s(n))$-qubit input state $\ket{\psi}$, first pad $\ket{\psi}$ with $s(n)$ 0-kets and then apply a Haar random $n$-qubit unitary. We also note that the choice of padding with $\ket{0}^{\ot s(n)}$ is arbitrary -- due to the invariance of the Haar measure we could use any fixed $s(n)$-qubit state for the padding.} In other words,  for any QPT algorithm $\algo A$, it holds that $$
    \vline\, \underset{k \sim \algo K}{\Pr}[\algo A^{V_k}(1^\lambda)=1] - \underset{V \sim \Haar(2^{n-s(n)},2^n)}{\Pr}[\algo A^{V}(1^\lambda) =1]  \,\vline \,\leq \, \negl(n)\,.
    $$ 
\end{itemize}
Note that, whenever we write $\mathcal{V}_n = \{V_k\}_{k \in \mathcal{K}}$, it is implicit that the key space $\mathcal{K}$ depends on the security parameter $n \in \N$, and that the length of each key $k \in \mathcal{K}$ is polynomial in $n$.
As in \cref{def:PRU}, we can also define a notion of non-adaptive security, but we do not do so here as we will show full adaptive security.
\end{definition}

Before stating the main result, we recall our PRI construction from \Cref{def:pri_construction}, which applied the $PF$ operator to input state after appending it with ancillas in the $\ket{+}^{s(n)}$ state where $s(n) = \omega(\log n)$. As mentioned in the proof overview earlier, we do not require the random Cliffords here (because the fixed $\ket{+}$-input already ensures that a suitable distinct string condition holds).

\begin{theorem} \label{thm:pri_adaptive}
The isometries defined in \Cref{def:pri_construction} are pseudorandom isometries with adaptive security. 
\end{theorem}

The proof of \cref{thm:pri_adaptive} proceeds by performing a reduction to the distinct subspace, then performing gate teleportation to reduce to the non-adaptive case with post-selection, which can then be analyzed with relative error $t$-designs. Before we execute this strategy to prove \cref{thm:pri_adaptive}, we need two auxiliary statements. The first one of these is a simple modification of \cref{thm:pfc-ensemble} and shows that on the distinct subspace, the $PF$ ensemble is a one-sided relative error $t$-design (\cref{def:relative-error-design}) for superpolynomial $t$.

\begin{lemma} \label{lem:rel_error_distinct}
Let $\Lambda$ be the projector onto the distinct string subspace of $\reg{A} \cong (\C^d)^{\ot t}$.
Let $\reg{\tilde E}$ be an arbitrary register and $\ket{\phi}_{\reg{A {\tilde E}}}$ a state such that $(\Lambda_{\reg A} \ot \id_{\reg{\tilde E}}) \ket{\phi}_{\reg{A {\tilde E}}} = \ket{\phi}_{\reg{A {\tilde E}}}$.
Then 
\begin{align*}
\mompft{\proj{\phi}_{\reg{A {\tilde E}}}} \leq (1 + O(t^2/d)) \momhaart{ \proj{\phi}_{\reg{A {\tilde E}}}}) \,,
\end{align*}
where as usual the twirling channels only act on register $\reg A$.
\end{lemma}

\begin{proof}
From \cref{lem:haar-twirl} and \cref{lem:pf-twirl}, we have the following explicit expressions: 
\begin{align*}
\mompft{\proj{\phi}_{\reg{A \tilde E}}} &= \sum_{\lambda \vdash t} \sigma_{\lambda} \otimes \Tr_{W_\lambda}[\ind_{P_\lambda} \proj{\phi}_{\reg{A \tilde E}}\ind_{P_\lambda}]\,,\\
\momhaart{\proj{\phi}_{\reg{A \tilde E}}} &= \sum_{\lambda \vdash t} \rho_{\lambda} \otimes \Tr_{W_\lambda}[\ind_{P_\lambda} \proj{\phi}_{\reg{A\tilde E}}\ind_{P_\lambda}] \,.
\end{align*}
Here, $\sigma_{\lambda} = \dfrac{\Lambda^{(\lambda)}_{W_\lambda}}{\Tr[\Lambda^{(\lambda)}_{W_\lambda}]} $ is the maximally mixed state on the subspace $\supp(\Lambda^{(\lambda)}_{W_\lambda}) \cap W_{\lambda}$ and $\rho_{\lambda} = \dfrac{\ind_{W_\lambda}}{\Tr[\ind_{W_\lambda}]} $ is the maximally mixed state on $W_{\lambda}$.

Since $\supp(\Lambda^{(\lambda)}_{W_\lambda}) \cap W_{\lambda}$ is a subspace of $W_\lambda$, we have the operator inequality $\Lambda^{(\lambda)}_{W_\lambda} \leq \ind_{W_\lambda}$.
This implies that for every $\lambda$,
\begin{align*}
\sigma_{\lambda} \leq \frac{\Tr[\id_{W_{\lambda}}]}{\Tr[\Lambda^{(\lambda)}_{W_\lambda}]} \rho_{\lambda} \leq \left( 1 + O(t^2/d) \right) \rho_{\lambda} \,.
\end{align*}
The second inequality uses \cref{cm:ratio} and the fact that $\rho_{\lambda} \geq 0$.

Since this holds for every $\lambda$ and taking the taking the tensor product with $\Tr_{W_\lambda}[\ind_{P_\lambda} \proj{\phi}_{\reg{A\tilde E}}\ind_{P_\lambda}]$ is an operator-monotone operation, the lemma follows.
\end{proof}

The second auxiliary lemma for the proof of \cref{thm:pri_adaptive} shows what happens in quantum gate teleportation when we do not use the ``correct'' resource state.
This is a straightforward calculation; the main difficulty is the notation required to distinguish the different registers from one another.
Both the statement and the proof of the lemma make heavy use of the unnormalised maximally entangled state between two registers $\reg{A} \cong \reg{A'}$: 
\begin{align*}
\ket{\Omega}_{\reg{AA'}} = \sum_{i = 1}^{|A|} \ket{i}_{\reg A} \ket{i}_{\reg A'}\,.
\end{align*}

\begin{lemma} \label{lem:general_gate_teleportation}
For $i = 1, \dots, t$, consider quantum registers $\reg{C_i}$ and $\reg{C_i'}$, all of size $d$.
Fix a collection of $d$-dimensional unitaries $A_i$ and a $d$-dimensional quantum state $\ket{\psi}_{\reg{C_0}}$.
Define the vectors
\begin{align*}
\ket{\Omega_{A_i}}_{\reg{C_i' C_{i+1}}} = ((A_i^\dagger)_{\reg{C_i'}} \ot \id_{\reg{C_{i+1}}}) \ket{\Omega}_{\reg{C_i' C_{i+1}}}\,.
\end{align*}
Define the following superoperator: 
\begin{multline*}
\cE(\phi_{\reg{C_1 C_1' \dots C_{t} C'_t}}) = \left( \bra{\Omega}_{C_0 C_1} \bra{\Omega_{A_1}}_{\reg{C_1' C_2}} \ot \dots \ot \bra{\Omega_{A_{t-1}}}_{\reg{C_{t-1}' C_{t}}} \ot (A_t)_{\reg{C'_t}} \right) \left( \proj{\psi}_{\reg{C_0}} \ot \phi_{\reg{\reg{C_1 C_1' \dots C_{t} C'_t}}} \right) \\ \left( \ket{\Omega}_{C_0 C_1} \ket{\Omega_{A_1}}_{\reg{C_1' C_2}} \ot \dots \ot \ket{\Omega_{A_{t-1}}}_{\reg{C_{t-1}' C_{t}}} \ot (A_t^\dagger)_{\reg{C'_t}} \right) \,.
\end{multline*}
Let $S \subseteq [d]^t$,  and overloading the notation define 
\begin{align*}
\ket{\Omega_S}_{\reg{C_1 C_1' \dots C_{t} C'_t}} = \sum_{(x_1, \dots, x_t) \in S} \ket{x_1}_{\reg{C_1}}\ket{x_1}_{\reg{C'_1}} \ot \dots \ot \ket{x_t}_{\reg{C_t}}\ket{x_t}_{\reg{C'_t}} \,.
\end{align*}
Then for any collection of $d$-dimensional 
unitaries $U_i$, 
\begin{align*}
&\cE \left( \Big( \id_{\reg{C_1}} \ot (U_1)_{\reg{C'_1}} \ot \dots  \id_{\reg{C_t}} \ot (U_t)_{\reg{C'_t}} \Big) \proj{\Omega_S} \Big( \id_{\reg{C_1}} \ot (U_1)_{\reg{C'_1}} \ot \dots  \id_{\reg{C_t}} \ot (U_t)_{\reg{C'_t}} \Big)^\dagger \right) \\
&= \mathrm{proj} \left(\sum_{(x_1, \dots, x_t) \in S} A_t U_t \proj{x_t} A_{t-1} U_{t-1} \proj{x_{t-1}} \dots A_1 U_1 \proj{x_1} \ket{\psi}_{\reg{C_0}}  \right) \,.
\end{align*}
Here, $\mathrm{proj}(\ket{\phi})$ is shorthand for $\proj{\phi}$.
\end{lemma}
\begin{proof}
It suffices to show that 
\begin{multline*}
\left( \bra{\Omega}_{C_0 C_1} \bra{\Omega_{A_1}}_{\reg{C_1' C_2}} \ot \dots \ot \bra{\Omega_{A_{t-1}}}_{\reg{C_{t-1}' C_{t}}} \ot (A_t)_{\reg{C'_t}} \right) \left( \ket{\psi}_{\reg{C_0}} \ot \Big( \id_{\reg{C_1}} \ot (U_1)_{\reg{C'_1}} \ot \dots  \id_{\reg{C_t}} \ot (U_t)_{\reg{C'_t}} \Big) \ket{\Omega_S}_{\reg{C_1 C_1' \dots C_{t} C'_t}} \right) \\
= \sum_{(x_1, \dots, x_t) \in S} A_t U_t \proj{x_t} A_{t-1} U_{t-1} \proj{x_{t-1}} \dots A_1 U_1 \proj{x_1} \ket{\psi}_{\reg{C_0}} \,.
\end{multline*}
Inserting the definitions of $\ket{\Omega_{A_i}}$ and $\ket{\Omega_S}$, we get that the l.h.s.~of the above equation is equal to 
\begin{align*}
&\sum_{\substack{(i_0, \dots, i_t) \in [d]^t \\ (x_1, \dots, x_t) \in S}}
\Big( \bra{i_0}_{\reg{C_0}} \bra{i_0}_{\reg{C_1}} \bra{i_1}_{\reg{C'_1}}  \bra{i_1}_{\reg{C_2}} \ot \dots \ot \bra{i_{t-1}}_{\reg{C'_{t-1}}}  \bra{i_{t-1}}_{\reg{C_t}} \Big) 
\Big( \ket{\psi}_{\reg{C_0}} \ot \id_{\reg{C_1}} \ot (A_1 U_1)_{\reg{C'_1}} \ot \dots \ot \id_{\reg{C_t}} \ot (A_t U_t)_{\reg{C'_t}} \Big) \\
& \qquad \qquad \qquad \Big( \ket{x_1}_{\reg{C_1}}\ket{x_1}_{\reg{C'_1}} \ot \dots \ot \ket{x_t}_{\reg{C_t}}\ket{x_t}_{\reg{C'_t}} \Big) \\
= & \sum_{\substack{(i_0, \dots, i_t) \in [d]^t \\ (x_1, \dots, x_t) \in S}} \delta_{i_0, x_1} \cdots \delta_{i_{t-1}, x_t}
\Big( \bra{i_0}_{\reg{C_0}} \bra{i_1}_{\reg{C'_1}} \ot \dots \ot  \bra{i_{t-1}}_{\reg{C'_{t-1}}} \Big) 
\Big( \ket{\psi}_{\reg{C_0}} \ot (A_1 U_1)_{\reg{C'_1}} \ot \dots \ot (A_t U_t)_{\reg{C'_t}} \Big) \\
& \qquad \qquad \qquad \Big( \ket{x_1}_{\reg{C'_1}} \ot \dots \ot \ket{x_t}_{\reg{C'_t}} \Big) \\
= & \sum_{(x_1, \dots, x_t) \in S} \braket{x_1 | \psi} \braket{x_2 | A_1 U_1 | x_1} \braket{x_3 | A_2 U_2 | x_2} \cdots \braket{x_t | A_{t-1} U_{t-1} | x_{t-1}} ( A_t U_t \ket{x_t}) \\
= & \sum_{(x_1, \dots, x_t) \in S} A_t U_t \proj{x_t} A_{t-1} U_{t-1} \proj{x_{t-1}} \dots A_1 U_1 \proj{x_1} \ket{\psi} \,.
\end{align*}
For the second line, we multiplied bras and kets on the non-primed systems $\reg{C_1, \dots, C_t}$ and used orthonormality of the basis kets.
\end{proof}

With the above two lemmas in hand, we can now prove \cref{thm:pri_adaptive}.

\begin{proof}[Proof of \cref{thm:pri_adaptive}]
First note that $V_k$ is efficient to implement by the same reasoning as in \cref{thm:pru_security}, and because the $\ket{+}^{\ot s(n)}$-state is efficient to prepare.

We need to prove that for any QPT algorithm $\algo A$, it holds that $$
    \vline\, \underset{k \sim \algo K}{\Pr}[\algo A^{V_k}(1^\lambda)=1] - \underset{V \sim \Haar(2^{n-s(n)},2^n)}{\Pr}[\algo A^{V}(1^\lambda) =1]  \,\vline \,\leq \, \negl(n)\,.
    $$ 
As in the proof of \Cref{thm:pru_security}, we can replace the pseudorandom isometry $V_k$ by a ``fully random'' isometry, where instead of using pseudorandom permutations and phases we sample uniformly random permutations and phases, i.e.~we sample from the $PF$ ensemble and then fix the last $s(n)$ qubits to be $\ket{+}$.
We denote an isometry sampled from this ensemble as $V \sim (PF)_{s(n)}$.
Since the algorithm $\cA$ is computationally efficient and we use quantum-secure pseudorandom permutations and functions, it follows that 
    $$
    \vline\, \underset{k \sim \algo K}{\Pr}[\algo A^{V_k}(1^\lambda)=1] - \underset{V \sim (PF)_{s(n)}}{\Pr}[\algo A^{V}(1^\lambda) =1]  \,\vline \,\leq \, \negl(n)\,.
    $$ 

By purifying the operations of $\cA$, for any fixed choice of isometry $V$ we can describe the algorithm's operation as follows: first prepare the state 
\begin{align}
    \ket{A(V)} \deq A_t (V \otimes I) A_{t-1} (V \otimes I) \cdots A_1 (V \otimes I) \ket{A_0} \label{eqn:adaptive_state}
\end{align}
(for some choice of unitaries $A_1, \dots, A_t$, which describe $\cA$'s actions in between querying the isometry, and an arbitrary initial state $\ket{A_0}$) and then perform a binary measurement $\{M, \id - M\}$, with ``$M$'' denoting the ``1''-outcome.
Then for any distribution $\cV$ of isometries, 
\begin{align*}
\Pr_{V \sim \cV}[\cA^V = 1] = \Tr[M \E_{V \sim \cV} \proj{A(V)}] \,.
\end{align*}
Therefore to prove the theorem, it suffices to show that
\begin{align}
\norm{\E_{V \sim (PF)_{s(n)}} \proj{A(V)} - \E_{V \sim \Haar(2^{n-s(n),2^n})} \proj{A(V)}}_1 = \negl(n) \,. \label{eqn:fri_security}
\end{align}

To proceed with the proof, we need to introduce some additional notation in \cref{eqn:adaptive_state}.
We can without loss of generality assume that $\cA$'s space is made up of the following registers: an $(n-s(n))$-qubit register $\reg A$, which is the input register for each application of the isometry; $s(n)$-qubit registers $\reg{B_1, \dots, B_t}$, where $\reg{B_i}$ is the additional output register produced by the $i$-the call to the isometry (i.e.~the $i$-call to the isometry is a map $\reg A \to \reg{A B_i}$); and an arbitrary workspace register $\reg R$.
The final state $\ket{A(V)}$ is then a state on registers $\reg{A R B_1 \dots B_t}$.

We can view the $i$-th call to the isometry $V: \reg A \to \reg{A B_i}$ as a call to a unitary $U: \reg{A B_i} \to \reg{A B_i}$, with the input on the $\reg{B_i}$ register fixed to $\ket{+^{s(n)}}_{\reg{B_i}}$.
This holds for both our construction of $V$, where we have the corresponding unitary $U = F P$, and a Haar random isometry $V$, where the corresponding unitary $U$ is Haar random.
Then we can view the additional $\ket{+^{s(n)}}_{\reg{B_i}}$-states as part of the input and write 
\begin{align*}
\ket{A(V)} \deq A_t (U_{\reg{A B_t}} \otimes \id_\reg{R B_{\setminus t}}) A_{t-1} (U_{\reg{A B_{t-1}}} \otimes \id_\reg{R B_{\setminus t-1}}) \cdots A_1 (U_{\reg{A B_1}} \otimes \id_{\reg{R B_{\setminus 1}}}) (\ket{A_0}_{\reg{A R}} \ot \ket{+^{s(n)}}_{\reg{B_1}} \ot \cdots \ot \ket{+^{s(n)}}_{\reg{B_t}}) \,.
\end{align*}
Here, $B_{\reg{\setminus i}}$ is shorthand for $\reg{B_1 \dots B_{i-1} B_{i+1} \dots B_t}$.

\paragraph{Reduction to distinct string inputs.}
As in the proof of \cref{thm:pfc-ensemble}, the next step is to restrict the inputs to the unitaries in $\ket{A(V)}$ to distinct strings.
The notion of the distinct string subspace is less clear for adaptive queries; we comment on what this means in the case of PRUs in \cref{sec:towards-adaptive}.
For the simpler case of isometries, we can simply project the fixed inputs $\ket{+^{s(n)}}_{\reg{B_1}} \ot \cdots \ot \ket{+^{s(n)}}_{\reg{B_t}}$ onto the distinct string subspace.

Formally, let $\Lambda_{\reg{B_1 \dots B_t}}$ be the projector onto the distinct string subspace of $\reg{B_1 \dots B_t}$ and let $\ket{+}_{\reg{B_1 \dots B_t}} := \ket{+^{s(n)}}_{\reg{B_1}} \ot \cdots \ot \ket{+^{s(n)}}_{\reg{B_t}}$ be the uniform superposition on registers $\reg{B_1 \dots B_t}$. Then, since $s(n) = \omega(\log n)$, it follows that
\begin{align*}
\norm{\ket{+}_{\reg{B_1 \dots B_t}} - \Lambda_{\reg{B_1 \dots B_t}} \ket{+}_{\reg{B_1 \dots B_t}}}_2 = \negl(n) \,,
\end{align*}
because the probability of observing a collision (i.e.~the probability of getting the same computational basis string in any two registers in $\reg{B_1 \dots B_t}$) when measuring $\ket{+}_{\reg{B_1 \dots B_t}}$ is negligible in $n$. This implies that 
\begin{align*}
\norm{\ket{A(V)} - \ket{A(V)_{\Lambda}}} = \negl(n) \,,
\end{align*}
where we defined 
\begin{align*}
\ket{A(V)_{\Lambda}} \deq A_t (U_{\reg{A B_t}} \otimes I) A_{t-1} (U_{\reg{A B_{t-1}}} \otimes I) \cdots A_1 (U_{\reg{A B_1}} \otimes I) (\ket{A_0}_{\reg{A R}} \ot \Lambda_{\reg{B_1 \dots B_t}} \ket{+}_{\reg{B_1 \dots B_t}}) \,.
\end{align*}
Consequently, it suffices to show that 
\begin{align}
\norm{\E_{V \sim (PF)_{s(n)}} \proj{A(V)_{\Lambda}} - \E_{V \sim \Haar} \proj{A(V)_{\Lambda}}}_1 = \negl(n) \,,\label{eqn:dis_security}
\end{align}
which then implies \cref{eqn:fri_security}.

\paragraph{Distinct state as output of gate teleportation.}

The next step is to write the state $\ket{A(V)_{\Lambda}}$ as the output of the gate teleportation map from \cref{lem:general_gate_teleportation}.
For this, we first observe that we can rewrite
\begin{multline}
\ket{A(V)_{\Lambda}} = \sum_{(y_1, \dots, y_t) \in \distinct(s(n), t)} A_t (U_{\reg{A B_t}} \proj{y_t}_{\reg{B_t}} \otimes I) A_{t-1} (U_{\reg{A B_{t-1}}} \proj{y_{t-1}}_{\reg{B_{t-1}}} \otimes I) \\ \cdots A_1 (U_{\reg{A B_1}} \proj{y_1}_{\reg{B_1}} \otimes I) (\ket{A_0}_{\reg{A R}} \ot \ket{+}_{\reg{B_1 \dots B_t}}) \,. \label{eqn:dis_yi}
\end{multline}
To make this look like an output state from \cref{lem:general_gate_teleportation}, in each step we can insert identities on all systems except the one that $\proj{y_i}$ is acting on.
More formally, we have that 
\begin{multline*}
\ket{A(V)_{\Lambda}} = \sum_{(x_1, \dots, x_t) \in S} A_t (U_{\reg{A B_t}} \otimes I) \proj{x_t} A_{t-1} (U_{\reg{A B_{t-1}}} \otimes I) \proj{x_{t-1}} \\ \cdots A_1 (U_{\reg{A B_1}} \otimes I) \proj{x_1} (\ket{A_0}_{\reg{A R}} \ot \ket{+}_{\reg{B_1 \dots B_t}})
\end{multline*}
for the following subset of strings (where $D = |R| + |A| + |B_1| + \dots + |B_t|$ is the total dimension of the space used by algorithm $\cA$):
\begin{align}
S = \Big\{ (x_1, \dots, x_t) \in [D]^t \;|\;
((x_1)_{[\reg{B_1}]}, \dots, (x_t)_{[\reg{B_t}]}) \in \distinct(s(n), t) \Big\} \,, \label{eqn:def_S}
\end{align}
where $(x_i)_{[\reg{B_i}]}$ denotes the substring of $x_i$ corresponding to register $B_i$.
In other words, $S$ contains tuples of strings; each string $x_i$ can be associated to a basis ket $\ket{x_i}_{\reg{RAB_1 \dots B_t}}$ on $\cA$'s total workspace; and $S$ includes all tuples such that when we only look at the part of $x_i$ that corresponds to register $\reg{B_i}$ (which is where $\proj{y_i}$ acts in \cref{eqn:dis_yi}), all of these substrings are distinct (since all $y_i$ in \cref{eqn:dis_yi} are distinct).
Using the shorthand $W_i = U_{\reg{AB_i}} \ot \id_{\reg{B_{\setminus i}} \ot \id_{\reg R}}$, we can write this more compactly as 
\begin{align*}
\ket{A(V)_{\Lambda}} = \sum_{(x_1, \dots, x_t) \in S} A_t W_t \proj{x_t} A_{t-1} W_{t-1} \proj{x_{t-1}} \cdots A_1 W_1 \proj{x_1} (\ket{A_0}_{\reg{A R}} \ot \ket{+}_{\reg{B_1 \dots B_t}})
\end{align*}

It then follows from \cref{lem:general_gate_teleportation} that 
\begin{align}
\proj{A(V)_{\Lambda}} = \cE \left( \Big( \id_{\reg{C_1}} \ot (W_1)_{\reg{C'_1}} \ot \dots  \id_{\reg{C_t}} \ot (W_t)_{\reg{C'_t}} \Big) \proj{\Omega_S} \Big( \id_{\reg{C_1}} \ot (W_1)_{\reg{C'_1}} \ot \dots  \id_{\reg{C_t}} \ot (W_t)_{\reg{C'_t}} \Big)^\dagger \right) \,, \label{eqn:output_of_tele}
\end{align}
where $\cE$ is the gate teleportation channel defined in \cref{lem:general_gate_teleportation} and $\reg{C'_i} \equiv \reg{C_i} \equiv \reg{AR B_1 \dots B_t}$ are copies of the total workspace.

\paragraph{Using the relative-error property on the distinct subspace.}
Observe that in the state 
\begin{align*}
\Big( \id_{\reg{C_1}} \ot (W_1)_{\reg{C'_1}} \ot \dots  \id_{\reg{C_t}} \ot (W_t)_{\reg{C'_t}} \Big) \proj{\Omega_S} \Big( \id_{\reg{C_1}} \ot (W_1)_{\reg{C'_1}} \ot \dots  \id_{\reg{C_t}} \ot (W_t)_{\reg{C'_t}} \Big)^\dagger
\end{align*}
from \cref{eqn:output_of_tele}, the different $U_{\reg{A B_i}}$ (which appear, tensored with identity, in $W_i$) act on distinct strings.
This is ensured by the fact that in the definition of the set $S$ (\cref{eqn:def_S}), we have the condition that the substrings $(x_i)_{[\reg{B_i}]}$ on the systems $\reg{B_i}$ (which is part of the input system to $U_{\reg{A B_i}}$) are distinct.
It therefore follows from \cref{lem:rel_error_distinct} that 
\begin{align*}
&\E_{U \sim PF} \Big( \id_{\reg{C_1}} \ot (W_1)_{\reg{C'_1}} \ot \dots  \id_{\reg{C_t}} \ot (W_t)_{\reg{C'_t}} \Big) \proj{\Omega_S} \Big( \id_{\reg{C_1}} \ot (W_1)_{\reg{C'_1}} \ot \dots  \id_{\reg{C_t}} \ot (W_t)_{\reg{C'_t}} \Big)^\dagger \\
\leq (1+O(t^2/d)) & \E_{U \sim \Haar} \Big( \id_{\reg{C_1}} \ot (W_1)_{\reg{C'_1}} \ot \dots  \id_{\reg{C_t}} \ot (W_t)_{\reg{C'_t}} \Big) \proj{\Omega_S} \Big( \id_{\reg{C_1}} \ot (W_1)_{\reg{C'_1}} \ot \dots  \id_{\reg{C_t}} \ot (W_t)_{\reg{C'_t}} \Big)^\dagger \,,
\end{align*}
where $W_i = U_{\reg{AB_i}} \ot I$ as before.

We can insert this into \cref{eqn:output_of_tele} and use the fact that the gate teleportation channel $\cE$ is manifestly completely positive to find that 
\begin{align}
\E_{V \sim (PF)_{s(n)}} \proj{A(V)_{\Lambda}} \leq (1 + O(t^2/d)) \E_{V \sim \Haar} \proj{A(V)_{\Lambda}} \,. \label{eqn:opineq_distinct}
\end{align}
To see how this implies \cref{eqn:dis_security}, note that using the variational definition of trace distance, there exists an $M$ such that $0 \leq M \leq \id$ and\footnote{We do not need absolute value signs because the expression in parentheses is the difference between two quantum states and therefore traceless. Consequently we can always replace $M \mapsto \id - M$ to switch the sign of the trace expression.} 
\begin{align*}
\text{l.h.s.~of \cref{eqn:dis_security}} &= \Tr[M \left( \E_{V \sim (PF)_{s(n)}} \proj{A(V)_{\Lambda}} - \E_{V \sim \Haar} \proj{A(V)_{\Lambda}} \right)] \\
&\leq O(t^2/d) \Tr[M \E_{V \sim \Haar} \proj{A(V)_{\Lambda}}] \\
&\leq O(t^2/d) \,.
\end{align*}
The first inequality follows by inserting \cref{eqn:opineq_distinct} and remembering that $M$ is positive semi-definite, so the map $X \mapsto \Tr[M X]$ is operator-monotone.
The second inequality uses the fact that the trace expression is at most 1, which holds because $M \leq \id$ and $\E_{V \sim \Haar} \proj{A(V)_{\Lambda}}$ is a quantum state.
This proves \cref{eqn:dis_security} and completes the proof.
\end{proof}

\subsection{Towards PRUs with adaptive security} \label{sec:towards-adaptive}

We briefly comment on possible ways to extend the proof of \cref{thm:pri_adaptive} to adaptively secure PRUs.
The reason why the current proof of \cref{thm:pri_adaptive} only works for isometries is the following.
We can analyse the state $\ket{A(V)_{\Lambda}}$ using the relative error property of the $PF$ ensemble (\cref{lem:rel_error_distinct}).
This part of the analysis is general and does not require fixing part of the input state to the unitary to $\ket{+}$.
The part of the analysis that limits us to isometries is relating the final state $\ket{A(V)}$ of an adaptive algorithm to the state $\ket{A(V)_{\Lambda}}$, where all the queries to $V$ (or rather its unitary extension $U = PF$) are restricted to distinct strings.
If we fix part of the input to the $\ket{+}$-state as we do in \cref{thm:pri_adaptive}, this fixed part of the input already ensures that the unitaries $U = PF$ are only queried on distinct strings (except with negligible weight).
This is the case because for $s(n) = \omega(\log n)$, the ``collision probability'' (i.e.~the probability of getting the same computational basis string) when measuring $\ket{+}^{\ot s(n)}$ is negligible in $n$.

To extend this to PRUs, we need to be able to ensure that no adaptive query algorithm queries the $PF$-ensemble twice with non-negligible weight on the same computational basis string.
The natural approach to this is to prepend the $PF$ ensemble with another ensemble of unitaries that sufficiently scrambles the input state to make sure that the $PF$ ensemble is not queried twice on the same computational basis string.
In the non-adaptive case, this role was played by a random Clifford unitary.
Unfortunately, it is unclear whether random Clifford unitaries still work for this purpose in the adaptive case.

More formally, we make the following conjecture.
\begin{conjecture} \label{conj:scrambling}
There exists an ensemble of efficient\footnote{In fact, it suffices if the ensemble is computationally indistinguishable from an ensemble of computationally efficient unitaries. Then we can replace the inefficient ensemble by the efficient one as in the proof of \cref{thm:pru_security}.} $n$-qubits unitaries $W \sim \cW$ such that for all $t = \poly(n)$ and all initial states $\ket{B_0}$ and sequences of unitaries $B_1, \dots, B_t$ on $n+\poly(n)$ qubits,
\begin{align*}
\norm{\E_{W \in \cW} \proj{B(W)} - \E_{W \in \cW} \proj{B(W)_\Lambda}}_1 = \negl(n) \,,
\end{align*}
where 
\begin{align*}
\ket{B(W)} &= B_t (W \ot \id_m) B_{t-1} (W \ot \id_m) \cdots  B_1 (W \ot \id_m) \ket{B_0} \,, \\
\ket{B(W)_\Lambda} &= \sum_{(x_1, \dots, x_t) \in \distinct(n,t)} B_t (\proj{x_t} W \ot \id_m) B_{t-1} (\proj{x_{t-1}} W \ot \id_m) \cdots  B_1 (\proj{x_{1}} W \ot \id_m) \ket{B_0} \,.
\end{align*}
\end{conjecture}
Assuming \cref{conj:scrambling}, it follows from the proof of \cref{thm:pri_adaptive} that the random unitary $PFW$, with $P$ and $F$ as before and $W \sim \cW$, is a PRU with adaptive security:
we can use the conjecture to perform the ``reduction to distinct string inputs'' in the proof of \cref{thm:pri_adaptive} (with $B_i = A_i (P F \ot I)$) and the rest of the proof goes through unchanged.
We leave it as an interesting open problem to prove \cref{conj:scrambling}.


\bibliographystyle{alpha}
\bibliography{main}

\end{document}